\documentclass{scrartcl} 
\usepackage{url}
\usepackage{graphicx}
\usepackage{fullpage}

\usepackage{authblk}

\usepackage{amsmath,amsthm,amssymb,enumerate}
\usepackage{algorithm,algpseudocode,float}

  \usepackage[text size=footnotesize]{todonotes}

\definecolor{ourblue}{RGB}{135,206,250}
\definecolor{ourgreen}{RGB}{0,100,0}
\definecolor{ourred}{RGB}{176,23,31}
\definecolor{lipicsyellow}{rgb}{0.99,0.78,0.07}

\usepackage[pagebackref]{hyperref}
\usepackage[numbers,sort]{natbib}
\usepackage{cleveref}

\usepackage{booktabs}
\usepackage{tabularx}
\usepackage{multirow}

\usepackage{paralist}

\usepackage{subcaption}
\usepackage{tikz}

\tikzstyle{vert}=[circle,inner sep=1.5,fill=white,draw,minimum size=.3cm]
\tikzstyle{vert2}=[circle,inner sep=1.5,fill=white,draw,minimum size=.2cm]

\pgfdeclarelayer{bg} 
\pgfsetlayers{bg,main}

\theoremstyle{plain}
\newtheorem{theorem}{Theorem}[section]

\newtheorem{definition}[theorem]{Definition}

\newtheorem{lemma}[theorem]{Lemma}
\newtheorem{observation}[theorem]{Observation}
\newtheorem{cor}[theorem]{Corollary}

\newtheorem{proposition}[theorem]{Proposition}

\theoremstyle{definition}
\newtheorem{claim}{Claim}
\newenvironment{proofofclaim}[1][Proof]{\noindent\emph{#1.} }{\ensuremath{\hfill\square}}

\crefname{lem}{Lemma}{Lemmata}
\crefname{proposition}{Proposition}{Proposition}
\crefname{theorem}{Theorem}{Theorems}
\crefname{cor}{Corollary}{Corollaries}
\crefname{claim}{Claim}{Claims}
\crefname{observation}{Observation}{Obervations}
\crefname{definition}{Definition}{Definitions}
\crefname{figure}{Figure}{Figures}
\crefname{section}{Section}{Sections}
\crefname{subsection}{Section}{Sections}

\newcommand{\TempColoring}{\textsc{Temporal Coloring}}
\newcommand{\SWTempColoring}{\textsc{SW-Temp.\ Coloring}}
\newcommand{\SWTempKColoring}{\textsc{SW-Temp.\ $k$-Coloring}}
\newcommand{\MinSWTempColoring}{\textsc{Minimum SW-Temp.\ Coloring}}

\newcommand{\tempKColoring}{temporal $k$-coloring}

\newcommand{\propTempKColoring}{proper \tempKColoring}
\newcommand{\DeltaTempKColoring}{sliding $\Delta$-window temporal $k$-coloring}
\newcommand{\propDeltaTempKColoring}{proper \DeltaTempKColoring}
\newcommand{\DeltaTempColoring}{sliding $\Delta$-window temporal coloring}
\newcommand{\propDeltaTempColoring}{proper \DeltaTempColoring}
\newcommand{\THREEFOURSAT}{\textsc{Exact (3,4)-SAT}}
\newcommand{\ONEINTHREESAT}{\textsc{Monotone Exactly 1-in-3 SAT}}

\DeclareOldFontCommand{\bf}{\normalfont\bfseries}{\mathbf}

\begin{document}
\title{Sliding Window Temporal Graph Coloring\footnote{GM and VZ are supported by the EPSRC grant EP/P020372/1, HM is supported by the DFG project MATE (NI 369/17). An extended abstract of this paper is accepted for publication at AAAI~'19~\cite{Mertzios2019aaai}.}}

\author[1]{George B.\ Mertzios}

\affil[1]{\small Department of Computer Science, Durham University, Durham, UK\\ \texttt{\{george.mertzios,viktor.zamaraev\}@durham.ac.uk}}

\author[2]{Hendrik~Molter}
\affil[2]{\small Algorithmics and Computational Complexity, Faculty~IV, TU Berlin, {Berlin, Germany}\\ \texttt{h.molter@tu-berlin.de}}

\author[1]{Viktor Zamaraev}

\date{ }

\maketitle
\begin{abstract}
Graph coloring is one of the most famous 
computational problems with applications in a wide range of areas such as planning and scheduling, resource allocation, and pattern matching. 
So far coloring problems are mostly studied on static graphs, which often stand in stark contrast to practice where data is inherently dynamic and subject to discrete changes over time. 
A temporal graph is a graph whose edges are assigned a set of integer time labels, indicating at which discrete time steps the edge is active. 
In this paper we present a natural temporal extension of the classical graph coloring problem. 
Given a temporal graph and a natural number $\Delta$, 
we ask for a coloring sequence for each vertex such that 
(i)~in every sliding time window of $\Delta$ consecutive time steps, in which an edge is active,
this edge is properly colored (i.e.~its endpoints are assigned two different colors) at least once during that time window, 
and (ii) the total number of different colors is minimized. 
This sliding window temporal coloring problem abstractly captures many realistic graph coloring scenarios in which the underlying network changes over time, 
such as dynamically assigning communication channels to moving agents. 
We present a thorough investigation of the computational complexity of this temporal coloring problem. 
More specifically, we prove strong computational hardness results, 
complemented by efficient exact and approximation algorithms. 
Some of our algorithms are linear-time fixed-parameter tractable 
with respect to appropriate parameters, 
while others are asymptotically almost optimal under the Exponential Time Hypothesis (ETH).
\end{abstract}

\section{Introduction}
\label{sec:intro}

A great variety of modern, as well as of traditional networks are dynamic in nature 
as their link availability changes over time. 
Just a few indicative examples of such inherently dynamic networks 
are information and communication networks, social networks, transportation networks, 
and several physical systems~\cite{Holme-Saramaki-book-13,michailCACM}. 
All these application areas share the common characteristic that the network structure, 
i.e.~the underlying graph topology, is subject to \emph{discrete changes over time}. 
In this paper, embarking from the foundational work of Kempe et al.~\cite{kempe}, 
we adopt a simple and natural model for time-varying networks, 
given by a graph with time-labels on its edges,
while the vertex set is fixed.

\begin{definition}[Temporal Graph]
\label{temp-graph-def} A \emph{temporal graph} is a pair $(G,\lambda)$,
where $G=(V,E)$ is an underlying (static) graph and $\lambda :E\rightarrow
2^{\mathbb{N}}$ is a \emph{time-labeling} function which assigns to every
edge of $G$ a set of discrete-time labels.
\end{definition}

For every edge $e\in E$ in the underlying graph $G$ of a temporal graph 
$(G,\lambda)$, $\lambda (e)$ denotes the set of time slots at which $e$ is \emph{active}. 
Due to their relevance and applicability in many areas, 
temporal graphs have been studied from various perspectives and under different names 
such as \emph{time-varying}~\cite{FlocchiniMS09,TangMML10-ACM,krizanc1}, 
\emph{dynamic}~\cite{GiakkoupisSS14,CasteigtsFloccini12}, 
\emph{evolving}~\cite{xuan,Ferreira-MANETS-04,clementi10}, 
and \emph{graphs over time} \cite{Leskovec-Kleinberg-Faloutsos07}. 
For a comprehensive overview on the existing models 
and results on temporal graphs from a (distributed) computing perspective see the 
surveys~\cite{michail2016introduction,latapy2017stream,CasteigtsFloccini12,flocchini1,flocchini2}.

The conceptual shift from static to temporal graphs imposes new challenges in algorithmic computation 
and complexity. 
Now the classical computational problems have to be appropriately redefined in the temporal setting 
in order to properly capture the notion of time. 
Motivated by the fact that, due to causality, information in temporal graphs can
``flow'' only along sequences of edges whose time-labels are increasing, 
most temporal graph parameters and optimization problems that have
been studied so far are based on the notion of temporal paths and other
``path-related'' notions, such as temporal analogues of distance, reachability, separators, diameter,
exploration, and 
centrality~\cite{akrida16,erlebach15,Mertzios13,michailTSP16,akridaTOCS17,enright2018deleting,ZschocheFMN18,FluschnikMNZ18}.

Recently only few attempts have been made to define and study ``non-path'' temporal graph problems. 
Motivated by the contact patterns among high-school students, 
Viard et al.~\cite{viardCliqueTCS16}, 
introduced $\Delta$-cliques, an extension of the concept of cliques to temporal graphs (see also~\cite{himmel17,BHMMNS18}). Chen et al.~\cite{CMSS2018} presented an extension of the cluster editing problem to temporal graphs.
Furthermore, Akrida et al.~\cite{Akrida-ICALP18} introduced the notion of temporal vertex cover, 
motivated by applications of covering problems in transportation and sensor networks. 
Temporal extensions of the classical graph coloring problem have also been previously studied by Yu et al.~\cite{yu2013algorithms} (see also~\cite{ghosal2015channel}) in the context of channel assignment in mobile wireless networks. 
In this problem, every edge has to be properly colored 
in every snapshot of the input temporal graph~$(G,\lambda)$, 
while the goal is to minimize some linear combination of the total number of colors used and the number of color re-assignments on the vertices~\cite{yu2013algorithms}. 
In this temporal coloring approach, the notion of time is only captured by the fact that the number of 
re-assignments affects the value of the target objective function, 
while the fundamental solution concept remains the same as in static graph coloring; that is,~\emph{every} individual (static) snapshot has to be properly colored. 
Using this, Yu et al.~\cite{yu2013algorithms} presented generic methods to adapt known algorithms 
and heuristics from static graph coloring to deal with their new objective function.
Other temporal extensions of the classical vertex and edge coloring problems have been recently studied by
Vizing~\cite{vizing2015coloring}. Vizing considered only temporal graphs of lifetime two, and in his problems
every object to color (vertex or edge) has to be colored in exactly one of the snapshots of the input temporal graph
in such a way that any two objects that are assigned the same color in the same snapshot are not adjacent in this snapshot. 
The goal of the problems is to minimize the total number of used colors.

In this paper we introduce and rigorously study a different, yet natural temporal extension of the classical graph coloring problem, 
called \textsc{Sliding Window Temporal Coloring} (for short, \SWTempColoring). 
In \SWTempColoring\ the input is a temporal graph~$(G,\lambda)$ and two natural numbers $\Delta$ and~$k$. 
At every time step~$t$, every vertex has to be assigned one color, under the following constraint: 
Every edge $e$ has to be properly colored  
at least once during \emph{every time window} of $\Delta$ consecutive time steps, 
and this must happen at a time step~$t$ in this window when $e$ is active. 
Now the question is whether there exists such a temporal coloring over the whole lifetime of the input temporal graph that uses at most $k$ colors. 
In contrast to the model of Yu et al.~\cite{yu2013algorithms}, the solution concept in \SWTempColoring\ is fundamentally different to that of static graph coloring as it takes into account the inherent dynamic nature of the network. 
Indeed, even to verify whether a given solution is feasible, it is not sufficient to just consider every snapshot independently.

Our temporal extension of the static graph coloring problem is motivated by applications in 
mobile sensor networks and in planning. 
Consider the following scenario: every mobile agent
broadcasts information over a specific communication channel while it listens on all \emph{other} channels. 
Thus, whenever two mobile agents are sufficiently close, they can exchange information only if they broadcast on different channels. 
We assume that agents can switch channels at any time. 
To ensure a high degree of information exchange, it makes sense to find a schedule of assigning broadcasting 
channels to the agents over time which minimizes the number of necessary channels, 
while allowing each pair of agents to communicate at least once within every small time window 
in which they are close to each other.

To further motivate the questions raised in this work, imagine an organization which, in order to ensure compliance with the national laws and the institutional policies, requires its employees to \emph{regularly} undertake special training that is relevant to 
their role within the organization. 
Such training requirements can be naturally grouped within training ``themes'', 
concerning --for example-- the General Data Protection Regulation (GDPR) of the EU for staff 
dealing with personal data
or equality and diversity issues when hiring new employees for Human Resources staff, etc). 
One reasonable organizational requirement for such a regular staff training is that every employee has to undertake all needed pieces of training at least once within every time-window of a specific length~$\Delta$ (e.g.~$\Delta = 12$ months). 
All training sessions are offered by experts in predefined ``training periods'' (e.g.~annually every January, May, and September), 
while each session takes a fixed amount of time to run (e.g.~a full day during the corresponding training period). 
This situation can be naturally modeled as a \emph{temporal} graph problem: 
\begin{inparaenum}[(i)]
\item each time slot $t$ represents a predefined ``training period''; 
\item each vertex $v$ denotes one of the themes that are offered for training by the organization; 
\item the different colors that a vertex $v$ can take at time slot~$t$ represent all different 
days in which the theme $v$ can be taught during the training period $t$;
\item an edge $\{u,v\}$ that is active at the time slot $t$ means that the themes $u$ and $v$ share at least one participant at the corresponding training period.
\end{inparaenum} 
Note that, since the training needs of specific staff members change over time, an edge between two themes $u$ and $v$ may repeatedly appear and disappear over time,
and thus the above
graph is temporal. 
If a participant is planned to undertake training on both themes $u,v$ at the same time slot~$t$, 
then these themes have to run at different days of the time slot~$t$, i.e.~$u$ and $v$ have to be assigned different colors at time $t$. 
In such a situation, it is natural for the organization to try to schedule all training sessions 
in such a way that the total \emph{duration} (i.e.~number of different colors) of every training period~$t$ 
never exceeds~$k$ different days, while simultaneously meeting all regular training requirements.

\subsection{Our Contribution.}

In this paper we introduce the problem \textsc{Sliding Window Temporal Coloring} (for short, \SWTempColoring) 
and we present a thorough investigation of its computational complexity. 
All our notation and the formal definition of the temporal problems that we study are presented in Section~\ref{sec:preliminaries}. 
First we investigate in Section~\ref{sec:TC} an interesting special case of \SWTempColoring, 
called \TempColoring, where the length $\Delta$ of the sliding time window is equal 
to the whole lifetime~$T$ of the input temporal graph. 
We start by proving in \cref{thm:TempKColoringNPhard} that \TempColoring\ is NP-hard even for $k=2$, 
and even when every time slot consists of one clique and isolated vertices. 
This is in wide contrast to the static coloring problem, where it can be decided in linear time 
whether a given (static) graph $G$ is 2-colorable, i.e.~whether $G$ is bipartite. 
On the positive side, we show in \cref{th:TCpolyKernel} that, given any input temporal graph $(G,\lambda)$ 
for \TempColoring\ with $n$ vertices and lifetime $T$, 
we can compute an equivalent instance $(G',\lambda')$ 
on the same vertices but with lifetime $T'\leq m$, where $m$ is the number of edges in 
the underlying graph $G$. Moreover we show that the new instance can be computed in polynomial time. 
Formally, \cref{th:TCpolyKernel} shows that \TempColoring\ 
admits a polynomial kernel when parameterized by the number $n$ of vertices of the input temporal graph. 
That is, we can efficiently preprocess any instance of \TempColoring\ 
to obtain an equivalent instance whose 
size only depends polynomially on the size of the underlying graph $G$ 
and not on the lifetime~$T$ of $(G,\lambda)$.

In Section~\ref{sec:SW-TC} and in the reminder of the paper we deal with the general version 
of \SWTempColoring, where the value of $\Delta$ is arbitrary. 
On the one hand, we show that the problem is hard even on very restricted special classes of input temporal graphs. 
On the other hand, assuming the Exponential Time Hypothesis (ETH), we give an asymptotically optimal exponential-time algorithm for \SWTempColoring\ whenever $\Delta$ is constant. 
Moreover we show how to extend it to get an algorithm
which runs in linear time if the number $n$ of vertices is constant. 
Note here that the \emph{size} of the input temporal graph also depends on its lifetime $T$ 
whose value can still be arbitrarily large, independently of $n$. 
Furthermore note that this assumption about $n$ being a constant can be also reasonable in 
practical situations; 
for example, in our motivation above about planning the training of staff in an organization, 
the value of $n$ equals the number of different ``training themes'' to be run, which can be expected to be rather small.

Finally we consider in Section~\ref{sec:SW-TC} an optimization variant of \SWTempColoring\ where the
number of colors is to be minimized. We give an approximation algorithm with an
additive error of 1 which runs in linear time on instances where the
underlying graph $G$ of the input temporal graph $(G,\lambda)$ has a constant-size vertex cover.
From a classification standpoint this is also optimal since the problem remains
NP-hard to solve optimally on temporal graphs where the underlying graph has a
constant-size vertex cover.

\section{Preliminaries and Notation}
\label{sec:preliminaries}

Given a (static) graph $G$, we denote by $V(G)$ and $E(G)$ the sets of its
vertices and edges, respectively. An edge between two vertices~$u$ and~$v$
of~$G$ is denoted by $\{u,v\}$, 
and in this case~$u$ and~$v$ are said to be \emph{adjacent} in~$G$. 
A \emph{complete graph} (or a \emph{clique}) is a graph where every pair of vertices is adjacent. 
The complete graph on $n$ vertices is denoted by $K_n$. 
For every $i,j\in \mathbb{N}$, where $i\leq j$, we let $[i,j]=\{i,i+1,\ldots ,j\}$ and $[j]=[1,j]$. 
Throughout the paper we consider temporal graphs with \emph{finite lifetime} $T$, that is, there is an maximum label assigned by $\lambda $ to an
edge of~$G$, called the \emph{lifetime} of $(G,\lambda )$; it is denoted by $%
T(G,\lambda )$, or simply by $T$ when no confusion arises. Formally, $%
T(G,\lambda )=\max \{t\in \lambda (e):e\in E\}$. 
We refer to each
integer $t\in [T]$ as a \emph{time slot} of $%
(G,\lambda )$. The \emph{instance} (or \emph{snapshot}) of $(G,\lambda )$ 
\emph{at time}~$t$ is the static graph $G_{t}=(V,E_{t})$, where $%
E_{t}=\{e\in E:t\in \lambda (e)\}$. 
If $E_t = \emptyset$, we call $G_t=(V,E_t)$ a \emph{trivial snapshot}.
For every subset $S\subseteq [T]$ of time slots, we denote by $(G,\lambda )|_S$ the restriction of $%
(G,\lambda )$ to the time slots in the set $S$. In particular, for the case where 
$S=[i,j]$ for some $i,j\in [T]$, where $i\leq j$, we have that 
$(G,\lambda )|_{[i,j]}$ is the sequence of the instances $G_{i},G_{i+1},\ldots ,G_{j}$.
We assume in the remainder of the paper that every edge of $G$ appears in at
least one time slot until $T$, namely $\bigcup_{t=1}^{T}E_{t}=E$.

In the remainder of the paper we denote by $n=|V|$ and $m=|E|$ the number of
vertices and edges of the underlying graph $G$, respectively, unless
otherwise stated. Furthermore, unless otherwise stated, we assume that the
labeling $\lambda $ is arbitrary, i.e.~$(G,\lambda )$ is given with an
explicit list of labels for every edge. That is, the \emph{size} of the
input temporal graph $(G,\lambda )$ is $O\left(
|V|+\sum_{t=1}^{T}|E_{t}|\right) =O(n+mT)$. In other cases, where $\lambda $
is more restricted, e.g.~if~$\lambda $ is periodic or follows another
specific temporal pattern, there may exist more succinct representations
of the input temporal graph.

For every $v\in V$ and every time slot $t$, we denote the \emph{appearance
of vertex} $v$ \emph{at time} $t$ by the pair $(v,t)$. 
That is, every vertex $v$ has $T$ different appearances (one for each time slot) during the
lifetime of $(G,\lambda)$. 
For every time slot $t\in [T]$ we denote by $V_{t}=\{(v,t) : v\in V\}$ the set of
all vertex appearances of $(G,\lambda)$ at the time slot $t$. 
Note that the set of all vertex appearances in $(G,\lambda)$ is the set $V\times [T] = \cup_{1\leq t\leq T} V_t$.

\subsection{\textsc{Temporal Coloring.}}
\label{subec:tc}

A \emph{temporal coloring} of a temporal graph $(G,\lambda)$ 
is a function $\phi : V\times [T] \rightarrow \mathbb{N}$, 
which assigns to every vertex appearance $(v,t)$ in $(G,\lambda)$ 
one color $\phi(v,t)\in\mathbb{N}$. 
For every time slot $t\in [T]$ we denote by~$\phi_{t}$ the restriction of $\phi$ to the 
vertex appearances at time slot $t$; then $\phi_{t}$ is referred to as the \emph{time slot coloring} 
for the time slot $t$. 
That is, $\phi_{t} : V \rightarrow \mathbb{N}$, such that 
$\phi_{t}(v)=\phi(v,t)$, for every $v\in V$. 
Furthermore, for simplicity of the presentation, we will refer to the temporal coloring $\phi$ as 
the ordered sequence $(\phi_{1}, \phi_{2}, \ldots, \phi_{T})$ of all its time slot colorings.
Let $e\in E$ be an edge of the underlying graph $G$. 
We say that an edge $e=\{u,v\}$ of the underlying graph $G$ is \emph{temporally properly colored} at time slot~$t$ 
if \begin{inparaenum}[(i)] \item $\phi_t(u) \neq \phi_t(v)$, 
and \item $t\in \lambda(e)$, i.e.~the edge $e$ is \emph{active} in the time slot $t$. \end{inparaenum} 
We now introduce the notion of a \emph{proper temporal coloring}
and the decision problem \textsc{Temporal Coloring}.

\begin{definition}
\label{def:temp-coloring}
Let $(G,\lambda)$ be a temporal graph with lifetime $T$, where $G=(V,E)$. 
A \emph{proper temporal coloring} of $(G,\lambda)$ is a temporal coloring 
$\phi=(\phi_{1}, \phi_{2}, \ldots, \phi_{T})$ such that every edge $e\in E$ is 
\emph{temporally properly colored} in at least one time slot $t\in \lambda(e)$. 
The \textit{size} of $\phi$ is the total number $|\phi|=|\bigcup_{i=1}^T \phi_i(V)|$ of colors used by $\phi$.
\end{definition}

\vspace{0,1cm} \noindent \fbox{ 
\begin{minipage}{0.97\textwidth}
 \begin{tabular*}{\textwidth}{@{\extracolsep{\fill}}lr} \textsc{Temporal Coloring}
& \\ \end{tabular*}
 
  \vspace{1.2mm}
{\bf{Input:}}  A temporal graph $(G,\lambda)$ with lifetime $T$ and an integer $k\in \mathbb{N}$.\\
{\bf{Question:}} Does there exist a proper temporal coloring $\phi=(\phi_{1}, \phi_{2}, \ldots, \phi_{T})$ of $(G,\lambda)$ using $|\phi|\leq k$ colors?
\end{minipage}} \vspace{0,3cm}

Note that \TempColoring\ is a natural extension of the problem \textsc{Coloring} to temporal graphs. 
In particular, \textsc{Coloring} is the special case of \TempColoring\ where the lifetime of the input temporal graph is $T=1$, and therefore \TempColoring\ is clearly NP-complete.

\subsection{\textsc{Sliding-Window Temporal Coloring.}}
\label{subec:sw-tc}

In the definition of a proper temporal coloring given in Definition~\ref{def:temp-coloring}, 
we require that every edge is properly colored at least once during the whole lifetime $T$ of the temporal 
graph~$(G,\lambda)$. However, in many real-world applications, where $T$ is expected to be arbitrarily large, 
we may need to require that every edge is properly colored more often, and in particular, at least once 
during \emph{every fixed period} $\Delta$ of time, regardless of how large the lifetime $T$ is.

Before we proceed with the formal definition of this problem, we first present some needed terminology. 
For every time slot $t\in \lbrack 1,T-\Delta +1]$, the \emph{$\Delta$-time window} $W_{t}=[t,t+\Delta -1]$ 
is the sequence of the $\Delta $
consecutive time slots $t,t+1,\ldots ,t+\Delta -1$. 
Furthermore we denote by $E[W_{t}]=\bigcup _{i\in W_{t}}E_{i}$ the union of all edges appearing 
at least once in the $\Delta$-time window~$W_{t}$. 
We are now ready to introduce the notion of a \emph{sliding $\Delta $-window temporal coloring} 
and the decision problem \SWTempColoring.

\begin{definition}
\label{def:sliding-Delta-temp-coloring}
Let $(G,\lambda)$ be a temporal graph with lifetime $T$, where $G=(V,E)$, and let $\Delta \leq T$. 
A \emph{proper sliding }$\Delta $\emph{-window temporal coloring} of $(G,\lambda)$ is a temporal coloring 
$\phi=(\phi_{1}, \phi_{2}, \ldots, \phi_{T})$ such that, 
for every $\Delta$-time window $W_{t}$ and for every edge $e\in E[W_{t}]$, 
$e$ is \emph{temporally properly colored} in at least one time slot $t\in W_{t}$.
The \textit{size} of $\phi$ is the total number $|\phi|=|\bigcup_{i=1}^T \phi_i(V)|$ of colors used by $\phi$.
\end{definition}

\vspace{0cm} \noindent \fbox{ 
\begin{minipage}{0.97\textwidth}
 \begin{tabular*}{\textwidth}{@{\extracolsep{\fill}}lr} \textsc{Sliding Window Temporal Coloring} \\ (\SWTempColoring) & \\ \end{tabular*}
 
  \vspace{1.2mm}
{\bf{Input:}}  A temporal graph $(G,\lambda)$ with lifetime $T$, and two integers $k\in\mathbb{N}$ and $\Delta \leq T$.\\
{\bf{Question:}} Does there exist a proper sliding $\Delta$-window temporal coloring $\phi=(\phi_{1}, \phi_{2}, \ldots, \phi_{T})$ of $(G,\lambda)$ using $|\phi|\leq k$ colors?
\end{minipage}} \vspace{0,3cm}

Whenever the parameter $\Delta $ is a fixed constant, 
we will refer to the above problem as the $\Delta $-\SWTempColoring\ 
(i.e.~$\Delta $ is now a part of the problem name). 
Moreover, whenever both $\Delta$ and $k$ are fixed constants, 
we will refer to the problem as the $\Delta $-\SWTempKColoring. 
Note that the problem \TempColoring\ defined above in this section is the 
special case of \SWTempColoring\ where $\Delta =T$, 
i.e.~where there is only one $\Delta$-window in the whole temporal graph. 
Another special case of \SWTempColoring\ is the problem $1$-\SWTempColoring, 
whose solution is obtained by iteratively solving the (static) \textsc{Coloring} problem on each of the $T$ static instances 
of $(G,\lambda)$. 
Thus $1$-\SWTempColoring\ fails to fully capture the time dimension in temporal graphs; 
in the remainder of the paper we will assume that $\Delta\geq 2$.

\subsection{Parameterized complexity.} 
We use standard notation and terminology from parameterized
complexity~\cite{CyganFKLMPPS15}. 
A \emph{parameterized problem} is a language $L\subseteq \Sigma^* \times \mathbb{N}$, where $\Sigma$ is a finite alphabet. We call the second component
the \emph{parameter} of the problem.
A parameterized problem is \emph{fixed-pa\-ram\-e\-ter tractable} (in the complexity class FPT)
if there is an algorithm that solves each instance~$(I, r)$ in~$f(r) \cdot |I|^{O(1)}$ time,
for some computable function $f$. 
A parameterized problem $L$ admits a \emph{polynomial kernel} if there is a polynomial-time algorithm that transforms each instance $(I,r)$ into an instance $(I', r')$ such that $(I,r)\in L$ if and only if $(I',r')\in L$ and~$|(I', r')|\le r^{O(1)}$.

\section{\TempColoring}\label{sec:TC}

In this section we investigate the complexity of \TempColoring. 
We start with our first hardness result.

\begin{theorem}\label{thm:TempKColoringNPhard}
\TempColoring\ is NP-hard even if $k=2$ and each snapshot has a clique and isolated vertices.
\end{theorem}
\begin{proof}
We reduce from \textsc{4-Coloring}. Let $G=(V,E)$ be an instance of \textsc{4-Coloring}. We construct a temporal graph $(G',\lambda)$ as follows:
$G'_1 = G'_2 = K_n$, $T = \binom{|V|}{2}-|E|+2$, and for every non-edge of $G$ there is exactly one time slot $i$ with $3 \leq i\leq T$ where only this edge is present. It is easy to check that every snapshot is a clique together with some isolated vertices. Before we proceed with proving the correctness of our reduction, we first prove the following auxiliary claim.

\begin{claim}\label{claim:bipartitespanning}
A graph $G$ has two bipartite spanning graphs that cover all edges of $G$ if
and only if $G$ is 4-colorable.
\end{claim}
\begin{proofofclaim}[Proof of Claim~\ref{claim:bipartitespanning}]
Assume a given graph $G=(V,E)$ has two bipartite spanning graphs $G_1=(V, E_1)$
and $G_2=(V, E_2)$ that cover all edges $E$ of $G$, i.e.~$E=E_1\cup E_2$. Let
$\phi_i : V \rightarrow \{1,2\}$ be the coloring of $G_i$ for $i\in\{1, 2\}$.
Then $\phi(v) := (\phi_1(v),\phi_2(v))$ is a 4-coloring for $G$: First note
that $|\phi(V)|\le 4$, since $\phi(V)\subseteq\{(1,1),(1,2),(2,1),(2,2)\}$. Now
let $\{v, w\}\in E$. By assumption we have that $\{v, w\}\in E_1\cup E_2$.
Assume $\{v, w\}\in E_1$ (the other case is analogous), then we have that
$\phi_1(v)\neq \phi_1(w)$. It follows that $\phi(v)\neq \phi(w)$.

It remains to show that if a given graph $G=(V,E)$ is 4-colorable, it has two
bipartite spanning graphs that cover all edges of $G$. Let $\phi : V
\rightarrow \{1,2,3,4\}$ be a 4-coloring for $G$. Let $E_1 := \{\{v, w\} \mid
\phi(v)\in\{1,2\} \wedge \phi(w)\in\{3,4\}\}$ and $E_2 := \{\{v, w\} \mid
\phi(v)\in\{1,3\} \wedge \phi(w)\in\{2,4\}\}$. It is easy to check that
$G_1=(V,E_1)$ is bipartite: One part is formed by vertices colored in 1 or
2 and the second part by vertices colored in 3 or 4. By definition $E_1$
does not contain edges between vertices from the same part. The argument
for $G_2=(V,E_2)$ is analogous. They are both spanning graphs since
$E_i\subseteq E$ for $i\in\{1,2\}$. It remains to show that $E=E_1\cup E_2$:
Let $\{v, w\}\in E$ then if $\{\phi(v),\phi(w)\}\in
\{\{1,3\},\{2,3\},\{1,4\},\{2,4\}\}$ then $\{v, w\}\in E_1$, otherwise (if
$\{\phi(v),\phi(w)\}\in \{\{1,2\},\{3,4\}\}$) we have that $\{v, w\}\in E_2$.
\end{proofofclaim}

\medskip

We are now ready to proceed with the proof of correctness of our reduction,
namely, show that the constructed temporal graph can be properly colored with 2
colors if and only if the input graph is 4-colorable.
\begin{itemize}
\item[($\Rightarrow$):] If $G$ is 4-colorable, we can use the 4-coloring of $G$
to 2-color $G'_1$ and $G'_2$ (cf.\ \Cref{claim:bipartitespanning}) and for every
edge that is not present in $G$, color it properly in the snapshot where it is present.
\item[($\Leftarrow$):] If $(G',\lambda)$ is properly colorable with $k=2$
colors, then all edges present in $G$ have to be properly colored either
in~$G_1$ or $G_2$, that is, the edges of $G$ can be covered by two bipartite
graphs, and hence we can 4-color $G$ (cf.\ \Cref{claim:bipartitespanning}).\qedhere
\end{itemize}
\end{proof}

With a more refined reduction, one can also show that \TempColoring\ remains hard even if each snapshot has very few edges.
\begin{theorem}\label{thm:vertexcoverhard}
\TempColoring\ is NP-hard for all $k\ge 2$ even if each snapshot has $O(k^2)$ edges.
\end{theorem}

\begin{proof}
We present a reduction from \THREEFOURSAT~\cite{tovey1984simplified} to
\textsc{Temporal 2-Coloring}. The reduction can be easily modified to a larger
number of colors, but we omit here the details. Recall that in \THREEFOURSAT\ we
are asked to decide whether a given Boolean formula~$\phi$ is satisfiable and $\phi$ is in conjunctive normal form where every clause has exactly three distinct literals and every variable
appears in exactly four clauses. Given a formula $\phi$ with $n$ variables and
$m$ clauses, we construct a temporal graph~$(G,\lambda)$ consisting of $(n+2m)$
snapshots, that is, one snapshot for each variable gadget and two snapshots for
each clause gadget. An illustration of the construction is given in
\cref{fig:reduction2}.
\begin{figure}[t]
\begin{center}
    \begin{subfigure}{0.45\textwidth}\centering
\scalebox{.85}{
    \begin{tikzpicture}[line width=1pt, scale=.6]
    \node[vert,minimum width=.72cm] (A1) at (0,0) {$w_1$}; 
    \node[vert,minimum width=.72cm] (A2) at (2,0) {$w_2$};
    \node[vert,minimum width=.72cm] (A3) at (4,0) {$w_3$};
    \node[vert,minimum width=.72cm] (A4) at (6,0) {$w_4$};

    \node[vert2,rectangle,fill=lightgray] (U1) at (0,2) {};
    \node[vert2,fill=lightgray] (U2) at (1,2) {};
    \node (U22) at (2,2) {$\ldots$};
    \node[vert2,fill=lightgray] (U3) at (3,2) {};
    \node[vert2,fill=lightgray] (U4) at (4,2) {};
    \node (U42) at (5,2) {$\ldots$};
    \node[vert2,fill=lightgray] (U5) at (6,2) {};
    
    \node[vert2,rectangle,fill=ourblue] (V1) at (-.4,-1.5) {};
    \node[vert2,rectangle,fill=ourblue] (V2) at (.4,-1.5) {};
    \node[vert2,rectangle,fill=ourblue] (V3) at (1.6,-1.5) {};
    \node[vert2,rectangle,fill=ourblue] (V4) at (2.4,-1.5) {};
    \node[vert2,rectangle,fill=ourblue] (V5) at (3.6,-1.5) {};
    \node[vert2,rectangle,fill=ourblue] (V6) at (4.4,-1.5) {};
    \node[vert2,rectangle,fill=ourblue] (V7) at (5.6,-1.5) {};
    \node[vert2,rectangle,fill=ourblue] (V8) at (6.4,-1.5) {};
    
    \draw (U1) -- (A1);
    \draw (U1) -- (A2);
    \draw (U1) -- (A3);
    \path (U1) edge[bend left=5] (A4);
    
    \node (X1) at (-1.5, -1.5) {$x_1$};
    
    \draw[line width=2pt] (V1) -- (A1);
    \draw[line width=2pt] (V2) -- (A1);
    \draw[line width=2pt] (V3) -- (A2);
    \draw[line width=2pt] (V4) -- (A2);
    \draw[line width=2pt] (V5) -- (A3);
    \draw[line width=2pt] (V6) -- (A3);
    \draw[line width=2pt] (V7) -- (A4);
    \draw[line width=2pt] (V8) -- (A4);
    
    \draw (V1) -- (V2);
    \draw (V3) -- (V2);
    \draw (V3) -- (V4);
    \draw (V5) -- (V4);
    \draw (V5) -- (V6);
    \draw (V7) -- (V6);
    \draw (V7) -- (V8);
    \path (V8) edge[bend left] (V1);
    
    \node[vert2,opacity=0] (V12) at (-.4,-2.5) {};
    \node[vert2,opacity=0] (V22) at (.4,-2.5) {};
    \node[vert2,opacity=0] (V32) at (1.6,-2.5) {};
    \node[vert2,opacity=0] (V42) at (2.4,-2.5) {};
    \node[vert2,opacity=0] (V52) at (3.6,-2.5) {};
    \node[vert2,opacity=0] (V62) at (4.4,-2.5) {};
    \node[vert2,opacity=0] (V72) at (5.6,-2.5) {};
    \node[vert2,opacity=0] (V82) at (6.4,-2.5) {};
    
    \node[vert2,opacity=0] (V13) at (-.4,-3.5) {};
    \node[vert2,opacity=0] (V23) at (.4,-3.5) {};
    \node[vert2,opacity=0] (V33) at (1.6,-3.5) {};
    \node[vert2,opacity=0] (V43) at (2.4,-3.5) {};
    \node[vert2,opacity=0] (V53) at (3.6,-3.5) {};
    \node[vert2,opacity=0] (V63) at (4.4,-3.5) {};
    \node[vert2,opacity=0] (V73) at (5.6,-3.5) {};
    \node[vert2,opacity=0] (V83) at (6.4,-3.5) {};

 \begin{pgfonlayer}{bg}   
    \draw[line width=2pt,opacity=0] (V1) -- (A1);
    \draw[line width=2pt,opacity=0] (V2) -- (A1);
    \draw[line width=2pt,opacity=0] (V3) -- (A2);
    \draw[line width=2pt,opacity=0] (V4) -- (A2);
    \draw[line width=2pt,opacity=0] (V5) -- (A3);
    \draw[line width=2pt,opacity=0] (V6) -- (A3);
    \draw[line width=2pt,opacity=0] (V7) -- (A4);
    \draw[line width=2pt,opacity=0] (V8) -- (A4);

    \path[line width=2pt,opacity=0] (V12) edge[bend right=5] (A1);
    \path[line width=2pt,opacity=0] (V22) edge[bend left=5] (A1);
    \path[line width=2pt,opacity=0] (V32) edge[bend right=5] (A2);
    \path[line width=2pt,opacity=0] (V42) edge[bend left=5] (A2);
    \path[line width=2pt,opacity=0] (V52) edge[bend right=5] (A3);
    \path[line width=2pt,opacity=0] (V62) edge[bend left=5] (A3);
    \path[line width=2pt,opacity=0] (V72) edge[bend right=5] (A4);
    \path[line width=2pt,opacity=0] (V82) edge[bend left=5] (A4);

    \path[line width=2pt,opacity=0] (V13) edge[bend left] (A1);
    \path[line width=2pt,opacity=0] (V23) edge[bend right] (A1);
    \path[line width=2pt,opacity=0] (V33) edge[bend left] (A2);
    \path[line width=2pt,opacity=0] (V43) edge[bend right] (A2);
    \path[line width=2pt,opacity=0] (V53) edge[bend left] (A3);
    \path[line width=2pt,opacity=0] (V63) edge[bend right] (A3);
    \path[line width=2pt,opacity=0] (V73) edge[bend left] (A4);
    \path[line width=2pt,opacity=0] (V83) edge[bend right] (A4);
    
    \path[line width=2pt,red,opacity=0] (V1) edge (V42);
    \path[line width=2pt,red,opacity=0] (V73) edge (V42);
    \path[line width=2pt,red,opacity=0] (V1) edge[bend right=55] (V73);
 \end{pgfonlayer}
    
    \end{tikzpicture}}
    \caption{Variable gadget.}\label{fig:reduction2:sub1}
    \end{subfigure}
    \begin{subfigure}{0.45\textwidth}\centering
\scalebox{.85}{
    \begin{tikzpicture}[line width=1pt, scale=.6]   
    \node[vert,minimum width=.72cm] (A1) at (0,0) {$w_1$}; 
    \node[vert,minimum width=.72cm] (A2) at (2,0) {$w_2$};
    \node[vert,minimum width=.72cm] (A3) at (4,0) {$w_3$};
    \node[vert,minimum width=.72cm] (A4) at (6,0) {$w_4$};

    \node[vert2,fill=lightgray] (U1) at (0,2) {};
    \node[vert2,fill=lightgray] (U2) at (1,2) {};
    \node (U22) at (2,2) {$\ldots$};
    \node[vert2,fill=lightgray] (U3) at (3,2) {};
    \node[vert2,fill=lightgray,rectangle] (U4) at (4,2) {};
    \node (U42) at (5,2) {$\ldots$};
    \node[vert2,fill=lightgray] (U5) at (6,2) {};
    
    \node (X1) at (-1.5, -1.5) {$x_1$};
    \node (X1) at (-1.5, -2.5) {$x_2$};
    \node (X1) at (-1.5, -3.5) {$x_3$};
    
    \node[vert2,fill=ourblue,rectangle] (V1) at (-.4,-1.5) {};
    \node[vert2,fill=ourblue,rectangle] (V2) at (.4,-1.5) {};
    \node[vert2,fill=ourblue] (V3) at (1.6,-1.5) {};
    \node[vert2,fill=ourblue] (V4) at (2.4,-1.5) {};
    \node[vert2,fill=ourblue] (V5) at (3.6,-1.5) {};
    \node[vert2,fill=ourblue] (V6) at (4.4,-1.5) {};
    \node[vert2,fill=ourblue] (V7) at (5.6,-1.5) {};
    \node[vert2,fill=ourblue] (V8) at (6.4,-1.5) {};
    
    \node[vert2,fill=lipicsyellow] (V12) at (-.4,-2.5) {};
    \node[vert2,fill=lipicsyellow] (V22) at (.4,-2.5) {};
    \node[vert2,fill=lipicsyellow,rectangle] (V32) at (1.6,-2.5) {};
    \node[vert2,fill=lipicsyellow,rectangle] (V42) at (2.4,-2.5) {};
    \node[vert2,fill=lipicsyellow] (V52) at (3.6,-2.5) {};
    \node[vert2,fill=lipicsyellow] (V62) at (4.4,-2.5) {};
    \node[vert2,fill=lipicsyellow] (V72) at (5.6,-2.5) {};
    \node[vert2,fill=lipicsyellow] (V82) at (6.4,-2.5) {};
    
    \node[vert2,fill=ourgreen!60] (V13) at (-.4,-3.5) {};
    \node[vert2,fill=ourgreen!60] (V23) at (.4,-3.5) {};
    \node[vert2,fill=ourgreen!60] (V33) at (1.6,-3.5) {};
    \node[vert2,fill=ourgreen!60] (V43) at (2.4,-3.5) {};
    \node[vert2,fill=ourgreen!60] (V53) at (3.6,-3.5) {};
    \node[vert2,fill=ourgreen!60] (V63) at (4.4,-3.5) {};
    \node[vert2,fill=ourgreen!60,rectangle] (V73) at (5.6,-3.5) {};
    \node[vert2,fill=ourgreen!60,rectangle] (V83) at (6.4,-3.5) {};
    
    \draw (U4) -- (A1);
    \draw (U4) -- (A2);
    \draw (U4) -- (A3);
    \draw (U4) -- (A4);
 \begin{pgfonlayer}{bg}   
    \draw[line width=2pt] (V1) -- (A1);
    \draw[line width=2pt] (V2) -- (A1);

    \path[line width=2pt] (V32) edge[bend right=5] (A2);
    \path[line width=2pt] (V42) edge[bend left=5] (A2);

    \path[line width=2pt] (V73) edge[bend left] (A4);
    \path[line width=2pt] (V83) edge[bend right] (A4);
    
    \path[line width=2pt,ourred] (V1) edge (V42);
    \path[line width=2pt,ourred] (V73) edge (V42);
    \path[line width=2pt,ourred] (V1) edge[bend right=55] (V73);
    \end{pgfonlayer}
    \end{tikzpicture}}
    \caption{First snapshot of clause gadget.}\label{fig:reduction2:sub2}
    \end{subfigure}
\end{center}
\caption{Illustration of the reduction from \THREEFOURSAT\ to \textsc{Temporal 2-Coloring} of the proof of \cref{thm:vertexcoverhard}. \cref{fig:reduction2:sub1} depicts the variable gadget for $x_1$. \cref{fig:reduction2:sub2} depicts the first snapshot of the clause gadget for clause~$(x_1\vee \lnot x_2 \vee x_3)$, where we have the first appearance of $x_1$ (blue), the second appearance of $x_2$ (yellow), and the fourth appearance of $x_3$ (green). The second snapshot of the clause gadget contains only the red triangle and is not depicted. In both figures vertices corresponding to the remaining variables are not depicted. The rectangular vertices constitute a vertex cover of size at most 9 in the respective snapshots. Thick edges are present in exactly two snapshots and thin edges are present in exactly one snapshot.}
\label{fig:reduction2}
\end{figure}
We start by adding four vertices $w_1$, $w_2$, $w_3$, and $w_4$ which will help
to encode the first, second, third, and fourth appearance of a variable.

\emph{Variable gadget:} For each variable $x_i$ with~$i\in [n]$ of $\phi$ we
create nine vertices $v_{x_i}^{(1)}$, $v_{x_i}^{(2)}$, $\ldots$,
$v_{x_i}^{(8)}$ (which we also refer to as ``the vertices corresponding to
$x_i$''), and $u_{x_i}$ and one new snapshot. In this new snapshot, we connect
$v_{x_i}^{(j)}$ with~$v_{x_i}^{((j\bmod 8)+1)}$ for all $j\in[8]$ and we connect~$v_{x_i}^{(2\ell-1)}$ and
$v_{x_i}^{(2\ell)}$ with $w_\ell$ for all~$\ell\in[4]$. Furthermore, we connect
$u_{x_i}$ with $w_1$, $w_2$, $w_3$, and $w_4$. It is easy to check that every
snasnapshot corresponsing to a variable contains 20 edges.

\emph{Clause gadget:} For each clause $c_i$ with~$1\le i\le m$ of $\phi$ we add
two new snapshots and one new vertex $u_{c_i}$. In the first new snapshot we
connect it with $w_1$, $w_2$, $w_3$, and $w_4$. Let~$x_j$ be a variable that
appears in clause $c_i$ and let this be the $\ell$th appearance of $x_j$ in
$\phi$. Then we connect $w_\ell$ with $v_{x_j}^{(2\ell-1)}$ and
$v_{x_j}^{(2\ell)}$ in the first new snapshot. Lastly, denote $x_{j_1}$,
$x_{j_2}$, and $x_{j_3}$ the three variables in $c_i$ appearing for the
$\ell_1$th, $\ell_2$th, and $\ell_3$th time, respectively, and let~$y_s=1$
if~$x_{j_s}$ appears non-negated in $c_i$ and $y_s=0$ otherwise. We pairwise
connect $v_{x_{j_1}}^{(2\ell_1-y_1)}$, $v_{x_{j_2}}^{(2\ell_2-y_2)}$, and
$v_{x_{j_3}}^{(2\ell_3-y_3)}$ in both the first and the second new snapshot, we
refer to these three vertices as ``the triangle corresponding to clause
$c_i$''. 
It is easy to check that every snapshot corresponding to a clause contains at
most 13 edges.

It is easy to check that the reduction can be computed in polynomial time. It
remains to show that $(G,\lambda)$ admits a proper temporal
2-coloring if and only if $\phi$ is satisfiable.

\emph{($\Rightarrow$):} Assume that we are given a satisfying assignment for
$\phi$. Then we construct a proper temporal 2-coloring for $G$
as follows. Let the two colors be red and blue. Whenever we do not specify the
color of vertices in a certain snapshot, those vertices can be colored
arbitrarily in that snapshot. In each snapshot, we color all vertices $u_{x_i}$
and~$u_{c_j}$ with~$i\in[n]$ and $j\in[m]$ red and vertices $w_1$, $w_2$,
$w_3$, and $w_4$ blue.

Now consider the snapshots corresponding to variable gadgets. If variable $x_i$
is set to true in the satisfying assignment for $\phi$, we color (in the
snapshot corresponding to the variable gadget for $x_i$) vertices
$v_{x_i}^{(2\ell-1)}$ yellow and vertices $v_{x_i}^{(2\ell)}$ blue for
$\ell\in[4]$. Otherwise we color the vertices exacly in the opposite way. This
leaves exactly four edges monochromatic in each snapshot corresponding to a
variable gadget. These will be colored properly in the four clause gadgets
corresponding to the four clauses where the corresponding variable appears.
 
Next, consider the snapshots corresponding to clause gadgets, in particular the
first snapshot corresponding to clause $c_i$. Let $x_1$, $x_2$, and $x_3$ be
the three variables appearing in $c_i$ and w.l.o.g.\ let $x_1$ be contained in
a literal that satisfies the clause and let that be the $\ell$th appearance of
$x_1$. If $x_1$ appears non-negated, we color $v_{x_1}^{(2\ell-1)}$ blue and
all other vertices corresponding to variables $x_1$, $x_2$, and $x_3$ yellow.
Otherwise, we color $v_{x_1}^{(2\ell)}$ blue and all other vertices
corresponding to variables $x_1$, $x_2$, and $x_3$ yellow. Since the literal
containing~$x_1$ satisfies clause $c_i$ we have that the edge between $w_\ell$
and $v_{x_1}^{(2\ell-1)}$ or $v_{x_1}^{(2\ell)}$, respectively, is colored
properly in the snapshot corresponding to the variable gadget of $x_1$. Hence
all edges between $w_1$, $w_2$, $w_3$, $w_4$ and vertices corresponding to
variables $x_1$, $x_2$, and $x_3$ are colored properly. Out of the edges that
form the triangle corresponding to $c_i$ in the snapshot corresponding to
clause $c_i$, exactly one is colored monochromatic.
We color the vertices of the triangle in the \emph{second} snapshot
corresponding to the variable clause of $c_i$ such that exactly that edge is
colored properly. It is easy to verify that this describes a proper temporal
2-coloring for $G$.

\emph{($\Leftarrow$):} Assume we are given a proper temporal
2-coloring for $(G,\lambda)$. Then we construct a satisfying assignment for
$\phi$ in the following way: We start with the observation that in a proper
temporal coloring, vertices $w_1$, $w_2$, $w_3$, and $w_4$ have the same colors
in each snapshot that corresponds to a variable gadget and in each first
snapshot corresponding to a clause gadget. Further, in each snapshot
corresponding to a variable gadget there is a cycle of size eight containing all vertices corresponding to the variable of that gadget. Let that
variable be $x_i$. Since all edges involved in this cycle only exists in this
one snapshot, there are exactly two ways to color this cycle. One of them
leaves the edges between $v_{x_i}^{(2\ell-1)}$ and $w_\ell$ monochromatic for
$\ell\in[4]$. The other way to color the cycle the inverse coloring and leaves
the edges between $v_{x_i}^{(2\ell)}$ and $w_\ell$ monochromatic for
$\ell\in[4]$.
In the first case, we set $x_i$ to false, and in the second case we set $x_i$
to true. We claim that this yields a satisfying assignment for $\phi$.
 
Assume for contradiction that it is not. Then there is a clause that is not
satisfied. Let that clause be $c_i$. We have that in a proper coloring, also
vertices $w_1$, $w_2$, $w_3$, and $w_4$ have the same colors in each first
snapshot that corresponds to a clause gadget. Consider the triangle
corresponding to clause $c_i$ in the first snapshot of the clause gadget of
$c_i$. We have that in a proper temporal coloring, this triangle can not be monochromatic, since otherwise, one of the three edges
is not properly colored in none of the snapshots of the temporal graph. Note that the
triangle edges only exist in the two snapshots corresponding to the clause
gadget of $c_i$ and in the second snapshot, not all three edges can be colored
properly. Hence, in the first snapshot of the clause gadget, at
least one of the vertices of the triangle corresponding to $c_i$ has a
different color than vertices $w_1$, $w_2$, $w_3$, and $w_4$. However, this
means that the corresponding variable is set to a truth value that satisfies
this clause---a contradiction.
\end{proof}

\subsection{Polynomial Kernel for \TempColoring.}
We prove that, given a temporal graph $(G,\lambda)$ 
for \TempColoring\ with $n$ vertices and $T$ time slots, 
we can efficiently compute an equivalent instance $(G',\lambda')$ with $T'\leq m$  time slots, 
where $m$ is the number of edges in~$G$. 
The main idea is that if we have sufficiently many time slots, every edge can be colored in its own time slot and any excess time slots can be removed (as they could be colored arbitrarily).
Formally, \cref{th:TCpolyKernel} shows that \TempColoring\ admits a polynomial kernel when parameterized by the number $n$ of vertices.

\begin{theorem}\label{th:TCpolyKernel}
 	Let $(G,\lambda)$ be a temporal graph of lifetime~$T$. 
 	Then there exists a temporal graph $(G',\lambda') = (G,\lambda)|_{S}$ for some $S \subseteq [T]$, $|S| \leq m=|E(G)|$ 
 	such that for any $k\ge 2$ we have that $(G,\lambda)$ admits a \propTempKColoring\ if and only if $(G',\lambda')$ admits a \propTempKColoring. Furthermore, $(G',\lambda')$ can be constructed in
 	$O(mT \sqrt{m+T})$ time.
\end{theorem}
\begin{proof} Let $(G,\lambda)$ be a temporal graph with lifetime $T$.
If $T \leq m=|E(G)|$, then we let $(G',\lambda') = (G,\lambda)$ and we are done. From now on we assume that $T>m$.
We define $B_{(G,\lambda)}$ to be the bipartite graph with two 
parts $E=E(G)$ and $[T]$, and the edge set $L = \{ (e,t) ~|~ e \in E, t \in [T], t \in \lambda(e) \}$, that is, $e \in E$ is adjacent
to $t \in [T]$ if and only if $e$ appears in time slot $t$ in $(G,\lambda)$.
Let $M = \{ (e_1,t_1), (e_2,t_2), \ldots, (e_s,t_s) \}$ be a maximum matching in $B = B_{(G,\lambda)}$.
We claim that $(G',\lambda') = (G,\lambda)|_{\{t_1, t_2, \ldots, t_s\}}$ admits a \propTempKColoring\ if and only if $(G,\lambda)$ admits a \propTempKColoring.

	Given a set $M' \subseteq M$ we denote $E_{M'} = \{ e ~|~ (e,t) \in M' \}$ and $S_{M'} = \{ t ~|~ (e,t) \in M' \}$.
	Let $M_1 \subseteq M$ be the set of edges such that every vertex in $E_{M_1}$ is reachable from a vertex in
	$\overline{E_{M}} = E \setminus E_M$ by an $M$-alternating path, i.e. a path whose edges belong alternately to 
	$M$ and not to $M$. Let $M_2 = M \setminus M_1$.

	We claim that $N_{B}(E_{M_1} \cup \overline{E_M}) = S_{M_1}$. First, a vertex $e \in E_{M_1} \cup \overline{E_M}$
	does not have a neighbour in $\overline{S_M} = [T] \setminus S_M$, as otherwise there would exist an $M$-augmenting
	path in $B$, contradicting the maximality of $M$. Also, a vertex $e \in E_{M_1} \cup \overline{E_M}$ is not adjacent
	to a vertex $t_j \in S_{M_2}$, as otherwise the corresponding matching neighbour $e_j$ of $t_j$ would be
	reachable by an $M$-augmenting path from a vertex in $\overline{E_M}$, which would contradict the fact that $(e_j,t_j)$
	belongs to $M_2$.

	The above claim means that those edges of $G$ that are in $E_{M_1} \cup \overline{E_M}$ appear,
	and therefore can be properly colored, only in time slots in $S_{M_1}$.
	Furthermore, all the edges in $E_{M_2}$ can be properly colored with 2 colors in slots in $S_{M_2}$:
	every edge $e \in E_{M_2}$ can be properly colored in the separate time slot $t$, where $(e,t) \in M_2$.
	This implies that that $(G',\lambda')$ admits a \propTempKColoring\ if and only if $(G,\lambda)$ admits a \propTempKColoring, as required.
	
	 We can construct the graph
 	 $B = B_{(G,\lambda)}$ in $O(mT)$ time, and find a maximum matching $M$
 	 in $B$ in $O(mT \sqrt{m+T})$ time~\cite{micali1980v}. 
\end{proof}

\section{\SWTempColoring}\label{sec:SW-TC}

In this section we thoroughly investigate the computational complexity of \SWTempColoring.

\subsection{NP-Hardness.}\label{subsec:SW-TC-hardness}

Before we present our main hardness result for \SWTempColoring, we start with the following intuitive observation. 

\begin{lemma}\label{obs:Delta-monotonicity-coloring}
For every fixed $\Delta$, the problem $(\Delta+1)$-\SWTempKColoring\ is computationally 
at least as hard as the problem $\Delta$-\SWTempKColoring. 
\end{lemma}
The main idea is the following: 
 given an algorithm $A$ for $(\Delta+1)$-\SWTempKColoring, we can use $A$ to also solve $\Delta$-\SWTempKColoring: we modify the instance of $\Delta$-\SWTempKColoring\ by inserting a trivial snapshot after every $\Delta$ consecutive snapshots, thus obtaining an equivalent instance of $(\Delta+1)$-\SWTempKColoring. 
\begin{proof}[Proof of \Cref{obs:Delta-monotonicity-coloring}]
To see the correctness of \Cref{obs:Delta-monotonicity-coloring}, 
let $\mathcal{A}$ be an algorithm that decides, for a fixed $k\ge 2$, whether an input 
temporal graph $(G,\lambda)$ admits a proper sliding $(\Delta+1)$-window temporal coloring 
of size at most~$k$. 
Then it is easy to verify that $\mathcal{A}$ can be also used to decide whether $(G,\lambda)$ 
admits a \propDeltaTempColoring\ of size at most~$k$, as follows: 
first we amend the input temporal graph $(G,\lambda)$ 
by inserting one trivial snapshot after every~$\Delta$ consecutive snapshots 
and then we apply $\mathcal{A}$ to the resulting temporal graph. 
\end{proof}

Since 1-\SWTempKColoring\ is equivalent to solving $T$ independent instances of static \textsc{$k$-Coloring}, 
\Cref{obs:Delta-monotonicity-coloring} demonstrates that
for any natural~$\Delta$, $\Delta$-\SWTempKColoring\ is at least as hard as \textsc{$k$-Coloring}. 
Thus, if \textsc{$k$-Coloring} is hard on some class~$\mathcal{X}$ of static graphs, 
then $\Delta$-\SWTempKColoring\ is also hard for the class of always $\mathcal{X}$ temporal graphs.

\Cref{thm:hardness,thm:underlyingVChard} below imply that the converse is \emph{not} true. 
In fact, there exist specific classes $\mathcal{X}$ of static graphs 
(graphs whose connected components have size $O(k)$ and graphs whose vertex cover has size $O(k)$, 
respectively) for which \textsc{$k$-Coloring} can be solved in \emph{linear} time 
(for every fixed $k\geq 2$), although 2-\SWTempKColoring\ is NP-hard on always $\mathcal{X}$ temporal graphs.

\begin{theorem}\label{thm:hardness}
Let $k\ge2$. Then 2-\SWTempKColoring\ is NP-hard, even if $T=3$ and:
\begin{compactitem}
  \item the underlying graph is $(k+1)$-colorable, 
  \item the underlying graph has a maximum degree in $O(k)$, and
  \item every snapshot has connected components with size $O(k)$.
\end{compactitem}
\end{theorem}
\begin{proof}
We present a reduction from \THREEFOURSAT~\cite{tovey1984simplified} to
\textsc{2-SW Temp.\ 2-Coloring}. The reduction can be easily modified to a larger
number of colors, but we omit here the details. Recall that in \THREEFOURSAT\ we are asked
to decide whether a given Boolean formula~$\phi$ is
satisfiable and $\phi$ is in conjunctive normal form where every clause has
exactly three distinct literals and every variable appears in exactly four
clauses. 
Given a formula
$\phi$ with $n$ variables and $m$ clauses, we construct a temporal
graph~$(G,\lambda)$ consisting of three snaphots, which we will refer to as
$G_1=(V,E_1)$, $G_2=(V,E_2)$, and $G_3=(V,E_3)$.
We construct the following variable gadgets and clause gadgets. An illustration
of the construction is given in \cref{fig:reduction}.

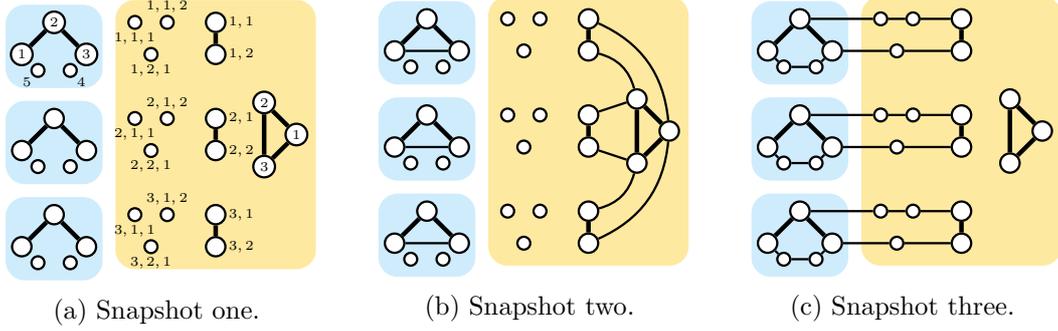
\begin{figure}[t]
\begin{center}
\begin{subfigure}{0.3\textwidth}\centering
\scalebox{.85}{
    \begin{tikzpicture}[line width=1pt, scale=.5]  
    \node[minimum width=1.5cm, minimum height=1.3cm, fill=ourblue!40, rounded corners=3mm] (A) at (0,.25) {}; 
    \node[minimum width=1.5cm, minimum height=1.3cm, fill=ourblue!40, rounded corners=3mm] (B) at (0,-2.75) {}; 
    \node[minimum width=1.5cm, minimum height=1.3cm, fill=ourblue!40, rounded corners=3mm] (C) at (0,-5.75) {}; 
    \node[minimum width=3.1cm, minimum height=4.2cm, fill=lipicsyellow!40, rounded corners=3mm] (G) at (5,-2.5) {}; 
    \node[vert] (A1) at (0,1) {\tiny$2$}; 
    \node[vert] (A2) at (-1,0) {\tiny$1$};
    \node[vert] (A3) at (1,0) {\tiny$3$};
    \node[vert2,label={[label distance=-5pt]225:\tiny$5$}] (A4) at (-.5,-.5) {};
    \node[vert2,label={[label distance=-5pt]315:\tiny$4$}] (A5) at (.5,-.5) {};
    \node[vert] (B1) at (0,-2) {};
    \node[vert] (B2) at (-1,-3) {};
    \node[vert] (B3) at (1,-3) {};
    \node[vert2] (B4) at (-.5,-3.5) {};
    \node[vert2] (B5) at (.5,-3.5) {};
    \node[vert] (C1) at (0,-5) {};
    \node[vert] (C2) at (-1,-6) {};
    \node[vert] (C3) at (1,-6) {};
    \node[vert2] (C4) at (-.5,-6.5) {};
    \node[vert2] (C5) at (.5,-6.5) {};
    
    \node[vert] (G1) at (6.5,-1.5) {\tiny$2$}; 
    \node[vert] (G2) at (7.5,-2.5) {\tiny$1$};
    \node[vert] (G3) at (6.5,-3.5) {\tiny$3$};    
    \node[vert,label={[label distance=-3pt]0:\tiny$1,1$}] (G11) at (5,1) {}; 
    \node[vert,label={[label distance=-3pt]0:\tiny$1,2$}] (G12) at (5,0) {};   
    \node[vert,label={[label distance=-3pt]0:\tiny$2,1$}] (G21) at (5,-2) {}; 
    \node[vert,label={[label distance=-3pt]0:\tiny$2,2$}] (G22) at (5,-3) {};   
    \node[vert,label={[label distance=-3pt]0:\tiny$3,1$}] (G31) at (5,-5) {}; 
    \node[vert,label={[label distance=-3pt]0:\tiny$3,2$}] (G32) at (5,-6) {};

    \node[vert2,label={[label distance=-3pt]270:\tiny$1,1,1$}] (A11) at (2.5,1) {};
    \node[vert2,label={[label distance=-3pt]90:\tiny$1,1,2$}] (A12) at (3.5,1) {};
    \node[vert2,label={[label distance=-3pt]270:\tiny$1,2,1$}] (A13) at (3,0) {};
    \node[vert2,label={[label distance=-3pt]270:\tiny$2,1,1$}] (B11) at (2.5,-2) {};
    \node[vert2,label={[label distance=-3pt]90:\tiny$2,1,2$}] (B12) at (3.5,-2) {};
    \node[vert2,label={[label distance=-3pt]270:\tiny$2,2,1$}] (B13) at (3,-3) {};
    \node[vert2,label={[label distance=-3pt]270:\tiny$3,1,1$}] (C11) at (2.5,-5) {};
    \node[vert2,label={[label distance=-3pt]90:\tiny$3,1,2$}] (C12) at (3.5,-5) {};
    \node[vert2,label={[label distance=-3pt]270:\tiny$3,2,1$}] (C13) at (3,-6) {};
    
    \draw[line width=2pt] (A2) -- (A1);
    \draw[line width=2pt] (A3) -- (A1);
    \draw[line width=2pt] (B2) -- (B1);
    \draw[line width=2pt] (B3) -- (B1);
    \draw[line width=2pt] (C2) -- (C1);
    \draw[line width=2pt] (C3) -- (C1);
    
    \draw[line width=2pt] (G1) -- (G2);
    \draw[line width=2pt] (G2) -- (G3);
    \draw[line width=2pt] (G1) -- (G3);
    \draw[line width=2pt] (G11) -- (G12);
    \draw[line width=2pt] (G21) -- (G22);
    \draw[line width=2pt] (G31) -- (G32);
    \end{tikzpicture}}
    \caption{Snapshot one.}\label{fig:reduction:sub1}
    \end{subfigure}
    \begin{subfigure}{0.3\textwidth}\centering
\scalebox{.85}{
    \begin{tikzpicture}[line width=1pt, scale=.5]   
    \node[minimum width=1.5cm, minimum height=1.3cm, fill=ourblue!40, rounded corners=3mm] (A) at (0,.25) {}; 
    \node[minimum width=1.5cm, minimum height=1.3cm, fill=ourblue!40, rounded corners=3mm] (B) at (0,-2.75) {}; 
    \node[minimum width=1.5cm, minimum height=1.3cm, fill=ourblue!40, rounded corners=3mm] (C) at (0,-5.75) {}; 
    \node[minimum width=3.1cm, minimum height=4.2cm, fill=lipicsyellow!40, rounded corners=3mm] (G) at (5,-2.5) {}; 
    \node[vert] (A1) at (0,1) {}; 
    \node[vert] (A2) at (-1,0) {};
    \node[vert] (A3) at (1,0) {};
    \node[vert2] (A4) at (-.5,-.5) {};
    \node[vert2] (A5) at (.5,-.5) {};
    \node[vert] (B1) at (0,-2) {};
    \node[vert] (B2) at (-1,-3) {};
    \node[vert] (B3) at (1,-3) {};
    \node[vert2] (B4) at (-.5,-3.5) {};
    \node[vert2] (B5) at (.5,-3.5) {};
    \node[vert] (C1) at (0,-5) {};
    \node[vert] (C2) at (-1,-6) {};
    \node[vert] (C3) at (1,-6) {};
    \node[vert2] (C4) at (-.5,-6.5) {};
    \node[vert2] (C5) at (.5,-6.5) {};
    
    \node[vert] (G1) at (6.5,-1.5) {}; 
    \node[vert] (G2) at (7.5,-2.5) {};
    \node[vert] (G3) at (6.5,-3.5) {};    
    \node[vert] (G11) at (5,1) {}; 
    \node[vert] (G12) at (5,0) {};   
    \node[vert] (G21) at (5,-2) {}; 
    \node[vert] (G22) at (5,-3) {};   
    \node[vert] (G31) at (5,-5) {}; 
    \node[vert] (G32) at (5,-6) {};

    \node[vert2] (A11) at (2.5,1) {};
    \node[vert2] (A12) at (3.5,1) {};
    \node[vert2] (A13) at (3,0) {};
    \node[vert2] (B11) at (2.5,-2) {};
    \node[vert2] (B12) at (3.5,-2) {};
    \node[vert2] (B13) at (3,-3) {};
    \node[vert2] (C11) at (2.5,-5) {};
    \node[vert2] (C12) at (3.5,-5) {};
    \node[vert2] (C13) at (3,-6) {};
    
    \draw[line width=2pt] (A2) -- (A1);
    \draw[line width=2pt] (A3) -- (A1);
    \draw (A3) -- (A2);
    \draw[line width=2pt] (B2) -- (B1);
    \draw[line width=2pt] (B3) -- (B1);
    \draw (B3) -- (B2);
    \draw[line width=2pt] (C2) -- (C1);
    \draw[line width=2pt] (C3) -- (C1);
    \draw (C3) -- (C2);
    
    \draw[line width=2pt] (G1) -- (G2);
    \draw[line width=2pt] (G2) -- (G3);
    \draw[line width=2pt] (G1) -- (G3);
    \draw[line width=2pt] (G11) -- (G12);
    \draw[line width=2pt] (G21) -- (G22);
    \draw[line width=2pt] (G31) -- (G32);
    \path (G11) edge[bend left] (G2);
    \path (G12) edge[bend left] (G1);
    \path (G21) edge (G1);
    \path (G22) edge (G3);
    \path (G31) edge[bend right] (G3);
    \path (G32) edge[bend right] (G2);
    \end{tikzpicture}}
    \caption{Snapshot two.}\label{fig:reduction:sub2}
    \end{subfigure}
    \begin{subfigure}{0.3\textwidth}\centering
\scalebox{.85}{
    \begin{tikzpicture}[line width=1pt, scale=.5]   
    \node[minimum width=1.5cm, minimum height=1.3cm, fill=ourblue!40, rounded corners=3mm] (A) at (0,.25) {}; 
    \node[minimum width=1.5cm, minimum height=1.3cm, fill=ourblue!40, rounded corners=3mm] (B) at (0,-2.75) {}; 
    \node[minimum width=1.5cm, minimum height=1.3cm, fill=ourblue!40, rounded corners=3mm] (C) at (0,-5.75) {}; 
    \node[minimum width=3.1cm, minimum height=4.2cm, fill=lipicsyellow!40, rounded corners=3mm] (G) at (5,-2.5) {}; 
    \node[vert] (A1) at (0,1) {}; 
    \node[vert] (A2) at (-1,0) {};
    \node[vert] (A3) at (1,0) {};
    \node[vert2] (A4) at (-.5,-.5) {};
    \node[vert2] (A5) at (.5,-.5) {};
    \node[vert] (B1) at (0,-2) {};
    \node[vert] (B2) at (-1,-3) {};
    \node[vert] (B3) at (1,-3) {};
    \node[vert2] (B4) at (-.5,-3.5) {};
    \node[vert2] (B5) at (.5,-3.5) {};
    \node[vert] (C1) at (0,-5) {};
    \node[vert] (C2) at (-1,-6) {};
    \node[vert] (C3) at (1,-6) {};
    \node[vert2] (C4) at (-.5,-6.5) {};
    \node[vert2] (C5) at (.5,-6.5) {};
    
    \node[vert] (G1) at (6.5,-1.5) {}; 
    \node[vert] (G2) at (7.5,-2.5) {};
    \node[vert] (G3) at (6.5,-3.5) {};    
    \node[vert] (G11) at (5,1) {}; 
    \node[vert] (G12) at (5,0) {};   
    \node[vert] (G21) at (5,-2) {}; 
    \node[vert] (G22) at (5,-3) {};   
    \node[vert] (G31) at (5,-5) {}; 
    \node[vert] (G32) at (5,-6) {};

    \node[vert2] (A11) at (2.5,1) {};
    \node[vert2] (A12) at (3.5,1) {};
    \node[vert2] (A13) at (3,0) {};
    \node[vert2] (B11) at (2.5,-2) {};
    \node[vert2] (B12) at (3.5,-2) {};
    \node[vert2] (B13) at (3,-3) {};
    \node[vert2] (C11) at (2.5,-5) {};
    \node[vert2] (C12) at (3.5,-5) {};
    \node[vert2] (C13) at (3,-6) {};
    
    \draw[line width=2pt] (A2) -- (A1);
    \draw[line width=2pt] (A3) -- (A1);    
    \draw (A2) -- (A4);
    \draw (A4) -- (A5);
    \draw (A5) -- (A3);
    \draw[line width=2pt] (B2) -- (B1);
    \draw[line width=2pt] (B3) -- (B1);   
    \draw (B2) -- (B4);
    \draw (B4) -- (B5);
    \draw (B5) -- (B3);
    \draw[line width=2pt] (C2) -- (C1);
    \draw[line width=2pt] (C3) -- (C1);   
    \draw (C2) -- (C4);
    \draw (C4) -- (C5);
    \draw (C5) -- (C3);
    
    \draw[line width=2pt] (G1) -- (G2);
    \draw[line width=2pt] (G2) -- (G3);
    \draw[line width=2pt] (G1) -- (G3);
    \draw[line width=2pt] (G11) -- (G12);
    \draw[line width=2pt] (G21) -- (G22);
    \draw[line width=2pt] (G31) -- (G32);
    
    \draw (G11) -- (A12);
    \draw (A12) -- (A11);
    \draw (A11) -- (A1);    
    \draw (G12) -- (A13);
    \draw (A13) -- (A3);
    
    \draw (G21) -- (B12);
    \draw (B12) -- (B11);
    \draw (B11) -- (B1);    
    \draw (G22) -- (B13);
    \draw (B13) -- (B3);
    
    \draw (G31) -- (C12);
    \draw (C12) -- (C11);
    \draw (C11) -- (C1);    
    \draw (G32) -- (C13);
    \draw (C13) -- (C3);
    \end{tikzpicture}}
    \caption{Snapshot three.}\label{fig:reduction:sub3}
    \end{subfigure}
\end{center}
\caption{Illustration of the reduction from \THREEFOURSAT\ to \textsc{2-SW Temp.\ 2-Coloring} of the proof of \cref{thm:hardness}. Vertices and edges in the yellow shaded areas (right) correspond to a clause gadget for clause~$(x_1\vee x_2 \vee x_3)$. Vertices and edges in the blue shaded areas (left) correspond to the variable gadgets for $x_1$, $x_2$, and $x_3$. Thick edges appear in every snapshot while thin edges only appear in one snapshot. In the first snapshot~(\subref{fig:reduction:sub1}), the superscripts of the vertices used in the proof of \cref{thm:hardness} are shown. To keep the figure clean, those are omitted in the illustrations for snapshots two (\subref{fig:reduction:sub2}) and three (\subref{fig:reduction:sub3}).}
\label{fig:reduction}
\end{figure}

\emph{Variable gadget:} For each variable $x_i$ with~$1\le i\le n$ of $\phi$ we create five vertices $v_{x_i}^{(1)}$, $v_{x_i}^{(2)}$, $v_{x_i}^{(3)}$, $v_{x_i}^{(4)}$, and $v_{x_i}^{(5)}$. The vertices $v_{x_i}^{(1)}$, $v_{x_i}^{(2)}$, and $v_{x_i}^{(3)}$ form a (not necessarily induced)~$P_3$ in every snapshot, that is $\{v_{x_i}^{(1)}, v_{x_i}^{(2)}\}\in E_t$ and $\{v_{x_i}^{(2)}, v_{x_i}^{(3)}\}\in E_t$ for all~$1\le t\le 3$. Furthermore, we connect $v_{x_i}^{(1)}$ and $v_{x_i}^{(3)}$ in the second snapshot, that is,~$\{v_{x_i}^{(1)}, v_{x_i}^{(3)}\}\in E_2$. Lastly, we create a full $C_5$ in snapshot three, that is, $\{v_{x_i}^{(3)}, v_{x_i}^{(4)}\}\in E_3$, $\{v_{x_i}^{(4)}, v_{x_i}^{(5)}\}\in E_3$, and~$\{v_{x_i}^{(1)}, v_{x_i}^{(5)}\}\in E_3$.

\emph{Clause gadget:} For each clause $c_i$ with~$1\le i\le m$ of $\phi$ we create a total of 18 vertices. We create vertices $v_{c_i}^{(1)}$, $v_{c_i}^{(2)}$, and $v_{c_i}^{(3)}$ and connect them to a triangle in every snapshot, that is, $\{v_{c_i}^{(1)}, v_{c_i}^{(2)}\}\in E_t$, $\{v_{c_i}^{(2)}, v_{c_i}^{(3)}\}\in E_t$, and $\{v_{c_i}^{(1)}, v_{c_i}^{(3)}\}\in E_t$ for all~$1\le t\le 3$. In this proof, we refer to these vertices as the \emph{core} of the clause gadget of clause $c_i$. Next, we add six vertices, which we refer to as the \emph{extension} of the core of the clause gadget of clause $c_i$. Let these vertices be called $v_{c_i}^{(1,1)}$, $v_{c_i}^{(1,2)}$, $v_{c_i}^{(2,1)}$, $v_{c_i}^{(2,2)}$, $v_{c_i}^{(3,1)}$, and $v_{c_i}^{(3,2)}$. We connect $v_{c_i}^{(j,1)}$ and~$v_{c_i}^{(j,2)}$ for all $1\le j\le 3$ in every snapshot, that is, $\{v_{c_i}^{(j,1)}, v_{c_i}^{(j,2)}\}\in E_t$ for all~$1\le j\le 3$ and for all~$1\le t\le 3$. In the second snapshot, we connect the extension and the core in the following way.
\begin{compactitem}
\item Edge $\{v_{c_i}^{(1,1)},v_{c_i}^{(1,2)}\}$ forms a $C_4$ with edge $\{v_{c_i}^{(2)}, v_{c_i}^{(1)}\}$, that is, $\{v_{c_i}^{(2)}, v_{c_i}^{(1,2)}\}\in E_2$ and $\{v_{c_i}^{(1)}, v_{c_i}^{(1,1)}\}\in E_2$.
\item Edge $\{v_{c_i}^{(2,1)},v_{c_i}^{(2,2)}\}$ forms a $C_4$ with edge $\{v_{c_i}^{(2)}, v_{c_i}^{(3)}\}$, that is, $\{v_{c_i}^{(2)}, v_{c_i}^{(2,1)}\}\in E_2$ and $\{v_{c_i}^{(3)}, v_{c_i}^{(2,2)}\}\in E_2$.
\item Edge $\{v_{c_i}^{(3,1)},v_{c_i}^{(3,2)}\}$ forms a $C_4$ with edge $\{v_{c_i}^{(1)}, v_{c_i}^{(3)}\}$, that is, $\{v_{c_i}^{(1)}, v_{c_i}^{(3,2)}\}\in E_2$ and $\{v_{c_i}^{(3)}, v_{c_i}^{(3,1)}\}\in E_2$.
\end{compactitem}
Lastly, we introduce nine auxiliary vertices that help to connect clause gadgets and variable gadgets. Let these vertices be called $v_{c_i}^{(j,1,1)}$, $v_{c_i}^{(j,1,2)}$, and $v_{c_i}^{(j,2,1)}$ for all $1\le j\le 3$. In the third snapshot, we connect the extension of the core and these auxiliary vertices in the following way. For all $1\le j\le 3$ we have that $\{v_{c_i}^{(j,1,1)}, v_{c_i}^{(j,1,2)}\}\in E_3$, $\{v_{c_i}^{(j,1,2)}, v_{c_i}^{(j,1)}\}\in E_3$, and $\{v_{c_i}^{(j,2,1)}, v_{c_i}^{(j,2)}\}\in E_3$.

\emph{Connection of variable and clause gadgets:} The clause gadgets and variable gadgets are connected in the third snapshot. Let clause $c_i=(\ell_{i,1}\vee\ell_{i,2}\vee\ell_{i,3})$ with $1\le i\le m$ have literals~$\ell_{i,1}$, $\ell_{i,2}$, and $\ell_{i,3}$. Let $x_{i,j}$ with $1\le i\le m$ and $1\le j\le 3$ be the variable of the $j$th literal in clause $c_i$. If $\ell_{i,j}=x_{i,j}$, then~$\{v_{x_{i,j}}^{(2)}, v_{c_i}^{(j,1,1)}\}\in E_3$ and~$\{v_{x_{i,j}}^{(3)}, v_{c_i}^{(j,2,1)}\}\in E_3$. If~$\ell_{i,j}=\lnot x_{i,j}$, then~$\{v_{x_{i,j}}^{(1)}, v_{c_i}^{(j,1,1)}\}\in E_3$ and~$\{v_{x_{i,j}}^{(2)}, v_{c_i}^{(j,2,1)}\}\in E_3$.

This completes the construction. Recall that $\Delta=2$ and $k=2$. It is easy
to check that the reduction can be computed in polynomial time. It remains to
show that $(G,\lambda)$ admits a proper sliding 2-window temporal 2-coloring if
and only if $\phi$ is satisfiable.

\emph{($\Rightarrow$):} Assume that we are given a satisfying assignment for
$\phi$. Then we construct a proper sliding 2-window temporal 2-coloring for $G$
as follows. We start coloring the second snapshot and then show that we can
color snapshots one and three in a way such that the complete coloring is a
proper sliding 2-window temporal 2-coloring. If a variable $x_i$ with $1\le
i\le n$ is set to true in the satisfying assignment, then we color the triangle
of the corresponding variable gadget in a way that leaves only 
edge~$\{v_{x_i}^{(1)}, v_{x_i}^{(2)}\}$ monochromatic. To be specific, assume
(for the remainder of this paragraph) we have colors yellow and blue, we color
vertices $v_{x_i}^{(1)}$ and $v_{x_i}^{(2)}$ in yellow and
vertices~$v_{x_i}^{(3)}$, $v_{x_i}^{(4)}$, and $v_{x_i}^{(5)}$ in blue. If
variable $x_i$ is set to false in the satisfying assignment, then we color the
triangle of the corresponding variable gadget in a way that leaves
edge~$\{v_{x_i}^{(2)}, v_{x_i}^{(3)}\}$ monochromatic. To be specific, we color
vertices $v_{x_i}^{(2)}$ and $v_{x_i}^{(3)}$ in yellow and vertices
$v_{x_i}^{(1)}$, $v_{x_i}^{(4)}$, and $v_{x_i}^{(5)}$ in blue. For each clause
$c_i$ with $1\le i\le m$ we choose one of its literals that satisfies the
clause. Let the $j$th literal with $1\le j\le3$ be a satisfying literal of
clause $c_i$ for the given assignment. Then we color the core of the
corresponding clause gadget in a way that leaves edge
$\{v_{c_i}^{(j)},v_{c_i}^{(j\bmod 3+1)}\}$ monochromatic. Note coloring the
core uniquely determines how we have to color the extension of the core since
the connecting edges are only present in the second snapshot and hence have to
be properly colored. The auxiliary vertices can be colored arbitrarily.

Now we show how to color snapshot one. For each variable $x_i$ with $1\le i\le n$, we color~$v_{x_i}^{(2)}$ in yellow and the remaining vertices of the corresponding gadget in blue. Note that this ensures that the edge which remains monochromatic in the second snapshot is properly colored in the first snapshot. For each clause $c_i$ with $1\le i\le m$ we color the core in a way that ensures that the edge which remains monochromatic in the second snapshot is properly colored in the first snapshot. We properly color all edges of the extension and the auxiliary vertices arbitrarily. It is not hard to see that now the first $\Delta$-window is properly colored.

Lastly, we show how to color the third snapshot. Note that for the variable
gadgets, the coloring in snapshot two determines (up to switching the colors)
how to color the variable gadgets in the third snapshot. This also determines how
to color the auxiliary vertices and the extension of the core in the third snapshot. 
This potentially leaves edges of the extension monochromatic. Note that in
the second snapshot, all extension edges are properly colored except the one
which, in the third snapshot, is connected to a variable that, in the given
assignment, satisfies the clause. It is straightforward to check that in this
case, this particular extension edge is properly colored in the third snapshot.
Lastly, the core is colored in a way that ensures that the edge that is colored
monochromatic in the second snapshot is colored properly in the third snapshot. 
It is easy to check that now the second $\Delta$-window is also properly colored.

\emph{($\Leftarrow$):} Assume we are given a proper sliding 2-window temporal
2-coloring for $(G,\lambda)$. Then we construct a satisfying assignment for
$\phi$ in the following way: Note that in the second snapshot each variable
gadget contains a triangle with exactly one monochromatic edge. The edge
$\{v_{x_i}^{(1)},v_{x_i}^{(3)}\}$ only exists in the second snapshot and hence
is colored properly by any proper sliding 2-window temporal 2-coloring. This
means that either edge $\{v_{x_i}^{(1)},v_{x_i}^{(2)}\}$ or
edge~$\{v_{x_i}^{(2)},v_{x_i}^{(3)}\}$ is colored monochromatic. If
$\{v_{x_i}^{(1)},v_{x_i}^{(2)}\}$ is colored monochromatic then we set~$x_i$ to
true, otherwise we set $x_i$ to false. We claim that this yields a satisfying
assignment for~$\phi$. Assume for contradiction that it is not. Then there is a
clause $c_j$ that is not satisfied. Without loss of generality, let $x_1$,
$x_2$, and $x_3$ be the variables appearing in $c_j$. Then in the third
snapshot, the clause gadget of $c_j$ is connected to the variable gadgets
of~$x_1$,~$x_2$, and~$x_3$. It is easy to check that in any proper sliding 2-window temporal 2-coloring, exactly one edge of the
extension of any clause gadget is colored monochromatic in the second snapshot,
hence this is also the case in the clause gadget of $c_j$. Without loss of
generality, let the monochromatically colored (in the second snapshot)
extension edge of the clause gadget of~$c_j$ be connected to the variable
gadget of $x_1$ in the third snapshot. It is easy to check that for the sliding
2-window temporal 2-coloring to be proper, the edge of the variable gadget of
$x_1$ that is connected to the clause gadget of $c_j$ in the third snapshot
needs to be colored properly in the second snapshot. By construction of
$(G,\lambda)$ this is a contradiction to $c_j$ not being satisfied by the
constructed assignment. 
\end{proof}

With small modifications to the reduction we get that \SWTempColoring\ remains hard under the following restrictions on the snapshots. 

\begin{cor}
\SWTempColoring\ is NP-hard for all $k\ge 2$, $\Delta\ge c$, and $T\ge \Delta+1$ for some constant $c$ even~if
\begin{compactitem}
\item every snapshot is a cluster graph, or
\item every snapshot has a dominating set of size one.
\end{compactitem}
\end{cor}
\begin{proof}[Proof Sketch]
Both modifications rely on the following construction. We can insert additional
vertices and edges to each of the three snapshots of the reduction presented in
the proof of \cref{thm:hardness}. Between snapshots two and three we add sufficiently new
snapshots containing exclusively new edges, such that all
new edges can be colored properly at least once if $\Delta$ is increased by the
number of new snapshots. Now all the new edges can be colored properly in the
newly inserted snapshots and the original construction of the reduction is not
affected.

To get the first property for all snapshots, we can add all edges that tranform
each connected component into a clique to the three original snapshots. Since
the components have constant size we only add a constant number of new edges per
component. Hence, we can add a constant number of new snapshots each containing
one new edge per component. The newly added snapshots are clearly cluster
graphs, hence we get the result.

To get the second property, we and one new universal vertex and one new snapshot
that contains all new edges. Clearly, now all snapshots have a dominating set of
size one.
\end{proof}

The reduction presented in the proof of \cref{thm:hardness} also yields a running time lower bound assuming the Exponential Time Hypothesis (ETH)~\cite{ImpagliazzoP01,ImpagliazzoPZ01}.

\begin{cor}\label{cor:ETHLB}
\SWTempColoring\ does not admit a $k^{o(n)\cdot f(T+k)}$-time algorithm for any computable function~$f$ unless ETH fails.
\end{cor}
\begin{proof}
First, note that any 3SAT formula with $m$ clauses can be transformed into an equisatisfiable \THREEFOURSAT\ formula with $O(m)$ clauses~\cite{tovey1984simplified}. The reduction presented in the proof of \cref{thm:hardness} produces an instance of \SWTempColoring\ with $n=O(m)$ vertices, $k=2$, and $T=3$. Hence an algorithm for \SWTempColoring\ with running time $k^{o(n)\cdot f(T+k)}$ for some computable function~$f$ would imply the existence of an $2^{o(m)}$-time algorithm for 3SAT. This is a contradiction to ETH~\cite{ImpagliazzoP01,ImpagliazzoPZ01}.
\end{proof}

\subsection{Optimal Exponential-Time Algorithm Assuming ETH.}

\label{subsec:SW-TC-positive}

In the following we give an exponential-time algorithm for $\Delta$-\SWTempColoring\ that asymptotically matches the running time lower bound given in \cref{cor:ETHLB} assuming the ETH.

We start with an auxiliary technical observation that allows us to combine partial colorings if they agree on their overlap.

\begin{observation}\label{lem:overlap}
Let $(G,\lambda)=(G_1, G_2, \ldots, G_i)$ and $(G',\lambda')=(G_j, G_{j+1}, \ldots, G_\ell)$ be two instances of $\Delta$-\SWTempKColoring\ defined over the same vertex set $V$ with $j+\Delta-1\le i\le\ell$. Let $\phi$ and $\psi$ be \DeltaTempKColoring s for $(G,\lambda)$ and $(G',\lambda')$, respectively, with the property that for all $v\in V$ and for all $j\le i^\star\le i$ we have that $\phi(v, i^\star)=\psi(v,i^\star)$.

We have that $\phi$ and $\psi$ are \propDeltaTempKColoring s if and only if $(\phi_1, \phi_2, \ldots \phi_i, \psi_{i+1} \ldots\psi_\ell)$ is a \propDeltaTempKColoring\ for $(G^\star,\lambda^\star)=(G_1, G_2, \ldots, G_\ell)$.
\end{observation}
\begin{proof}
The first direction is obvious.
Assume for contradiction that $(\phi_1, \phi_2, \ldots \phi_i, \psi_{i+1} \ldots\psi_\ell)$ is \emph{not} a \propDeltaTempKColoring\ for $(G_1, G_2, \ldots, G_\ell)$. Then there is a $\Delta$-window $W$ and an edge $e\in E$ that appears at least once in $W$ but is never properly colored in $W$. However, it is easy to check that $W$ is completely contained in $[i]$ or $[j,\ell]$ and hence $\phi$ or $\psi$ color the whole $\Delta$-window $W$. Since by assumption, both $\phi$ and $\psi$ are \propDeltaTempKColoring s, there being an edge that exists in $W$ and is not properly colored is a contradicion.
\end{proof}

Now we are ready to describe an exponential time algorithm for \SWTempColoring. The main idea is to enumerate all partial \propDeltaTempColoring s for time windows of size $2\Delta$ and then check whether we can combine them to a \propDeltaTempColoring\ for the whole temporal graph using \cref{lem:overlap}.

\begin{theorem}\label{thm:singleexp}
\SWTempColoring\ can be solved in $O(k^{4\Delta\cdot n}\cdot T)$ time.
\end{theorem}
\begin{proof}
For the sake of simplicity, we assume that $T$ is divisible by $\Delta$. The general case can be proven alike.
We give the following algorithm for the problem:
\begin{compactenum}
  \item For $2\Delta$-windows $W_i = [i\Delta+1, (i+2)\Delta]$ for $i\in\{0,1,\ldots,T/\Delta-2\}$, enumerate all partial \propDeltaTempColoring s $\phi_{W_i}$, where each trivial snapshot is colored in some fixed but arbitrary way\footnote{This is an important trick that allows us to use this algorithm for the FPT result in \cref{thm:fpt-size-n}.}. 
  \item Create a \emph{directed acyclic graph} (DAG) with $\phi_{W_i}$ as vertices and connect $\phi_{W_i}$ and $\phi_{W_{i+1}}$ with a directed arc if the two proper $\Delta$-temporal colorings agree on the overlapping part.
  \item Create a source vertex $s$ and connect it to all $\phi_{W_1}$ with a directed arc and we create a sink vertex $t$ and add a directed arc from all $\phi_{W_{T/\Delta-2}}$ to it.
  \item If there is a path from $s$ to $t$, answer YES, otherwise NO.
\end{compactenum}
The running time is dominated by checking whether~$s$ and~$t$ are connected in the last step of the algorithm. This can be done e.g.\ by a BFS on the constructed DAG. The DAG has at most $k^{2\Delta\cdot n}\cdot T$ vertices and at most $k^{4\Delta \cdot n} \cdot T$ edges.
\end{proof}

\subsection{Fixed-Parameter Tractability.} Next, we show how to extend the algorithm presented in \cref{thm:singleexp} to achieve linear time fixed-parameter tractability with respect to the number $n$ of vertices. The main idea is to reduce the number of non-trivial snapshots in each $\Delta$-window.

\begin{theorem}\label{thm:fpt-size-n}
\SWTempColoring\ can be solved in $O(T)$ time if $n$ is a constant.
\end{theorem}
\begin{proof}
We present a preprocessing step to reduce the number of non-trivial snapshots in any $\Delta$-window and then use the algorithm of \cref{thm:singleexp} to solve the problem.

The reduction rule is based on the observation that if some snapshot appears at least $n^2$ times in a $\Delta$-window, then the edges of this snapshot can be properly colored with 2 colors within the $\Delta$-window. In other words, all but $n^2$ copies of the snapshot in the $\Delta$-window are redundant for optimal coloring and each of them could be replaced by the trivial snapshot. When implementing this idea one should take care to guarantee that replacing a snapshot by the trivial one does not reduce the number of copies of the snapshot in other $\Delta$-windows which contain at most $n^2$ copies of the snapshot.

Formally the reduction rule is as follows. Since the number of different snapshots is at most $2^{{n \choose 2}} \leq 2^{n^2}$,
by the pigeonhole principle if $\Delta > 2 \cdot 2^{n^2} \cdot n^2$, then in every $\Delta$-window there exists a snapshot that appears more than $2n^2$ times in that $\Delta$-window. For every such a snapshot that contains at least one edge, we replace by the trivial snapshot one of its "middle" copies, that is, one that has at least $n^2$ copies appearing earlier and $n^2$ copies that appear later in the $\Delta$-window. This reduction rule guarantees that every $\Delta$-window that contains the modified snapshot also contains at least $n^2$ copies of the original snapshot appearing either earlier or later in the $\Delta$-window.

The reduction rule can be applied exhaustively by
linearly sweeping over all $\Delta$-windows once in the following way. For each different graph (snapshot) 
we store a list of occurrences and update these lists every time
we move the $\Delta$-window by one. Having these lists, it is straightforward to count the occurrences and replace the middle ones by trivial snapshots. When we move the $\Delta$-window, we just have to update two lists: the one of the graph that enters the $\Delta$-window and the one of the graph that leaves. This requires a lookup table of size
$2^{\binom{n}{2}}\le 2^{n^2}$ but takes only linear time in~$T$.
Note that after this procedure, every $\Delta$-window contains at most $2\cdot 2^{n^2}\cdot n^2$ non-trivial snapshots.

Now we apply the algorithm of \cref{thm:singleexp}. 
Note that after the reduction step the number of \emph{non-trivial} snapshots in every $\Delta$-window depends only on $n$. Furthermore, since we can assume that $k \leq n$, the number of colorings that are enumerated in Step 1 of the algorithm in \cref{thm:singleexp} is bounded by a function of $n$. This completes the proof.
\end{proof}

The FPT result of \Cref{thm:fpt-size-n} is complemented by the following theorem, in which we 
exclude the possibility of a polynomial-sized kernel for \SWTempColoring\ 
with respect to the number $n$ of vertices. 
This comes in contrast to the existence of a polynomial-sized kernel for \TempColoring\ with respect to $n$ 
(cf.~\Cref{th:TCpolyKernel}).

\begin{proposition}\label{thm:nopk}
\SWTempColoring\ does not admit a polynomial-sized kernel with respect to the number $n$ of vertices for all $\Delta\ge 2$ and $k\ge 2$ unless NP $\subseteq$ coNP/poly.
\end{proposition}
 We need the following notation for the proof.
An equivalence
relation~$R$ on the instances of some problem~$L$ is a
\emph{polynomial equivalence relation} if
\begin{compactenum}[(i)]
 \item one can decide for each two instances in time polynomial in their sizes whether they belong to the same equivalence class, and
 \item for each finite set~$S$ of instances, $R$ partitions the set into at most~$(\max_{x \in S} |x|)^{O(1)}$ equivalence classes.  
\end{compactenum}

An \emph{AND-cross-composition} of a problem~$L\subseteq \Sigma^*$ into a
parameterized problem~$P$ (with respect to a polynomial equivalence
relation~$R$ on the instances of~\(L\)) is an algorithm that takes
$\ell$ $R$-equivalent instances~$x_1,\ldots,x_\ell$ of~$L$ and
constructs in time polynomial in $\sum_{i=1}^\ell |x_i|$ an instance
$(x,k)$ of~\(P\) such that
\begin{compactenum}[(i)]
\item $k$ is polynomially upper-bounded in $\max_{1\leq i\leq \ell}|x_i|+\log(\ell)$ and 
\item $(x,k)$ is a yes-instance of $P$ if and only if $x_{\ell'}$ is a yes-instance of $L$ for every $\ell'\in \{1,\ldots,\ell\}$. 
\end{compactenum}

If an NP-hard problem~\(L\) AND-cross-composes into a parameterized
problem~$P$, then~$P$ does not admit a polynomial-size kernel, unless NP $\subseteq$ coNP/poly~\cite{bodlaender2014kernelization, CyganFKLMPPS15},
which would cause a collapse of the polynomial-time hierarchy to the third
level.

\begin{proof}[Proof of \Cref{thm:nopk}]
We provide an AND-cross-composition from \THREEFOURSAT~\cite{tovey1984simplified}. Recall that in \THREEFOURSAT\ we are asked to decide whether a given Boolean formula~$\phi$ is satisfiable and $\phi$ is in conjunctive normal form where every clause has exactly three distinct literals and every variable appears in exactly four clauses.
We define relation~$R$: Two instances $\phi$ and $\psi$ are equivalent under~$R$ if and only if the number of variables and the number of clauses is the same in both formulas. Clearly, $R$ is a polynomial equivalence relation. 

Now let $\phi_1,\ldots,\phi_\ell$ be $R$-equivalent instances of \THREEFOURSAT. Arbitrarily number all variables and clauses of all formulas. For each $\phi_i$ with $i\in[\ell]$ we construct an instance of \SWTempColoring\ as defined in the proof of \cref{thm:hardness} (for a illustration see \cref{fig:reduction}) with the only difference that we add a fourth and fifth snapshot both of which are copies of the first snapshot (\cref{fig:reduction:sub1}). Now we put all constructed temporal graphs next to each other in temporal order, that is, if $(G_1^{(i)},G_2^{(i)}, \ldots, G_5^{(i)})$ is the graph constructed for $\phi_i$, then the overall temporal graph is $(G_1^{(1)},G_2^{(1)}, \ldots, G_5^{(1)},G_1^{(2)},G_2^{(2)}, \ldots, G_5^{(2)},\ldots,$ $G_1^{(\ell)},G_2^{(\ell)}, \ldots, G_5^{(\ell)})$. Here, the vertex set stays the same. We identify the vertices with their names according to the numbering of the variables and clauses of the formulas. Further, we set $\Delta=2$ and $k=2$.

This instance can be constructed in polynomial time and the number of vertices is linearly bounded in the size of the formulas, hence~$|V|$ is polynomially upper-bounded by a maximum size of an input instance. Furthermore, it is easy to check that the two extra copies of the first snapshot in the construction (\cref{fig:reduction:sub1}) allows to go from an arbitrary proper coloring of snapshot $G_4^{(i)}$ to $G_1^{(i+1)}$ for any $i\in[\ell-1]$. It follows from the proof of \cref{thm:hardness}, that the constructed \SWTempColoring\ instance is a 
yes-instance if and only if for every~$i\in [\ell]$ formula $\phi_i$ is satisfiable.

Since \THREEFOURSAT\ is NP-hard~\cite{tovey1984simplified} and we AND-cross-composed it into \SWTempColoring{} with $\Delta=2$ and $k=2$ parameterized by $n=|V|$, the result follows.
\end{proof}

\subsection{Structural Graph Parameters and Approximation.} Finally, we investigate the possibility to structurally improve the fixed-parameter tractability result by replacing the parameter $n$ with a smaller parameter. We answer this negatively by showing that $\Delta$-\SWTempKColoring\ remains NP-hard even if the underlying graph has a constant-size vertex cover, which is a fairly large structural parameter. 

\begin{theorem}
\label{thm:underlyingVChard}
Let $k\geq 2$. Then 2-\SWTempKColoring\ is NP-hard, even if the vertex cover number of the underlying graph is at most $2k+13$.
\end{theorem}
\begin{proof}
We present a reduction from \ONEINTHREESAT~\cite{schaefer1978complexity} to
2-\SWTempKColoring. The reduction can be easily modified to a larger
number of colors, but we omit here the details. In \ONEINTHREESAT\ we are given a
collection of triples (clauses) of variables and the task is to determine whether there is an
assignment of truth values to variables such that each clause contains exactly
one variable that is set to true.
Given an instance $I$ of \ONEINTHREESAT\ with $n$ variables and $m$ clauses, we
construct a temporal graph $(G,\lambda)$ with~$T=4m$~snapshots in the following
way. The construction is visualized in \cref{fig:reductionunderlyingvc}.
\begin{figure}[t]
\begin{center}
\begin{subfigure}{0.24\textwidth}\centering
\scalebox{.8}{
    \begin{tikzpicture}[line width=1pt, scale=.45]
    \clip (-5,-9) rectangle (4, 5); 
    \node[vert, fill=ourblue] (A1) at (-1,2) {}; 
    \node[vert, fill=lipicsyellow] (A2) at (1,2) {};
    \node[vert, fill=lipicsyellow] (A3) at (-1,4) {};
    \node[vert, fill=ourblue] (A4) at (1,4) {};

    \node[vert2,fill=lipicsyellow] (U1) at (-3,0) {};
    \node[vert2,fill=ourblue] (U2) at (-2,0) {};
    \node (U22) at (-1,0) {$\ldots$};
    \node[vert2,fill=ourblue] (U3) at (0,0) {};
    \node[vert2,fill=ourblue] (U4) at (1,0) {};
    \node (U42) at (2,0) {$\ldots$};
    \node[vert2,fill=lipicsyellow] (U5) at (3,0) {};
    
    \node[vert] (B1) at (-2,-2) {}; 
    \node[vert] (B2) at (0,-2) {};
    \node[vert] (B3) at (2,-2) {};    
    \node[vert, fill=lipicsyellow] (B4) at (-2,-4) {}; 
    \node[vert, fill=lipicsyellow] (B5) at (0,-4) {};
    \node[vert, fill=lipicsyellow] (B6) at (2,-4) {};    
    \node[vert, fill=ourblue] (B7) at (-2,-6) {}; 
    \node[vert, fill=ourblue] (B8) at (0,-6) {};
    \node[vert, fill=ourblue] (B9) at (2,-6) {};    
    \node[vert] (B10) at (-2,-8) {}; 
    \node[vert] (B11) at (0,-8) {};
    \node[vert] (B12) at (2,-8) {};
    
    \node[vert, fill=lipicsyellow] (C) at (-3,-7) {};
    
    \path[line width=2pt]  (A1) edge (A3);
    \path[line width=2pt]  (A4) edge (A3);
    \path[line width=2pt]  (A4) edge (A2);
    \draw (A1) -- (U1);
    \draw (A1) -- (U2);
    \draw (A1) -- (U3);
    \draw (A1) -- (U4);
    \draw (A1) -- (U5);
    \draw (A2) -- (U1);
    \draw (A2) -- (U2);
    \draw (A2) -- (U3);
    \draw (A2) -- (U4);
    \draw (A2) -- (U5);
    
    \path  (B7) edge (B4);
    \path  (B8) edge (B5);
    \path  (B9) edge (B6);
     
     \path[line width=2pt, bend left=60,draw=white]  (C) edge (A1);  
    
    \path (B1) edge (B2);
    \path (B2) edge (B3);
    \path[bend left]  (B3) edge (B1);
    
    \path (B10) edge (B11);
    \path (B11) edge (B12);
    \path[bend left]  (B12) edge (B10);

    \end{tikzpicture}}
    \caption{Type 1 snapshot.}\label{fig:reductionvc:sub0}
    \end{subfigure}
    \begin{subfigure}{0.24\textwidth}\centering
\scalebox{.8}{
    \begin{tikzpicture}[line width=1pt, scale=.45]
    \clip (-5,-9) rectangle (4, 5); 
    \node[vert, fill=ourblue] (A1) at (-1,2) {}; 
    \node[vert, fill=lipicsyellow] (A2) at (1,2) {};
    \node[vert, fill=lipicsyellow] (A3) at (-1,4) {};
    \node[vert, fill=ourblue] (A4) at (1,4) {};

    \node[vert2,fill=ourblue] (U1) at (-3,0) {};
    \node[vert2,fill=lipicsyellow] (U2) at (-2,0) {};
    \node (U22) at (-1,0) {$\ldots$};
    \node[vert2,fill=lipicsyellow] (U3) at (0,0) {};
    \node[vert2,fill=lipicsyellow] (U4) at (1,0) {};
    \node (U42) at (2,0) {$\ldots$};
    \node[vert2,fill=ourblue] (U5) at (3,0) {};
    
    \node[vert, fill=ourblue] (B1) at (-2,-2) {}; 
    \node[vert, fill=ourblue] (B2) at (0,-2) {};
    \node[vert, fill=lipicsyellow] (B3) at (2,-2) {};    
    \node[vert, fill=lipicsyellow] (B4) at (-2,-4) {}; 
    \node[vert, fill=lipicsyellow] (B5) at (0,-4) {};
    \node[vert, fill=ourblue] (B6) at (2,-4) {};    
    \node[vert, fill=ourblue] (B7) at (-2,-6) {}; 
    \node[vert, fill=ourblue] (B8) at (0,-6) {};
    \node[vert, fill=ourblue] (B9) at (2,-6) {};    
    \node[vert] (B10) at (-2,-8) {}; 
    \node[vert] (B11) at (0,-8) {};
    \node[vert] (B12) at (2,-8) {};
    
    \node[vert, fill=lipicsyellow] (C) at (-3,-7) {};
    
    \path[line width=2pt]  (A1) edge (A2);
    \draw (A1) -- (U1);
    \draw (A1) -- (U2);
    \draw (A1) -- (U3);
    \draw (A1) -- (U4);
    \draw (A1) -- (U5);
    \draw (A2) -- (U1);
    \draw (A2) -- (U2);
    \draw (A2) -- (U3);
    \draw (A2) -- (U4);
    \draw (A2) -- (U5);
    
    \path[line width=2pt,ourred]  (B1) edge (U2);
    \path[line width=2pt,ourred]  (B2) edge (U3);
    \path[line width=2pt,ourred]  (B3) edge (U5);
    \path[line width=2pt]  (B1) edge (B4);
    \path[line width=2pt]  (B2) edge (B5);
    \path[line width=2pt]  (B3) edge (B6);
    
    \path  (B7) edge (B4);
    \path  (B8) edge (B5);
    \path  (B9) edge (B6);
    
    \path[line width=2pt]  (C) edge (B7);
    \path[line width=2pt]  (C) edge (B8);
    \path[line width=2pt]  (C) edge (B9);  
     
    \path[line width=2pt, bend left=60]  (C) edge (A1);  
    
    \path (B1) edge (B2);
    \path (B2) edge (B3);
    \path[bend left]  (B3) edge (B1);
    
    \path (B10) edge (B11);
    \path (B11) edge (B12);
    \path[bend left]  (B12) edge (B10);

    \end{tikzpicture}}
    \caption{Type 2 snapshot.}\label{fig:reductionvc:sub1}
    \end{subfigure}
    \begin{subfigure}{0.24\textwidth}\centering
\scalebox{.8}{
    \begin{tikzpicture}[line width=1pt, scale=.45]
    \clip (-5,-9) rectangle (4, 5); 
    \node[vert, fill=ourblue] (A1) at (-1,2) {}; 
    \node[vert, fill=lipicsyellow] (A2) at (1,2) {};
    \node[vert, fill=lipicsyellow] (A3) at (-1,4) {};
    \node[vert, fill=ourblue] (A4) at (1,4) {};

    \node[vert2,fill=lipicsyellow] (U1) at (-3,0) {};
    \node[vert2,fill=ourblue] (U2) at (-2,0) {};
    \node (U22) at (-1,0) {$\ldots$};
    \node[vert2,fill=ourblue] (U3) at (0,0) {};
    \node[vert2,fill=ourblue] (U4) at (1,0) {};
    \node (U42) at (2,0) {$\ldots$};
    \node[vert2,fill=lipicsyellow] (U5) at (3,0) {};
    
    \node[vert] (B1) at (-2,-2) {}; 
    \node[vert] (B2) at (0,-2) {};
    \node[vert] (B3) at (2,-2) {};    
    \node[vert, fill=lipicsyellow] (B4) at (-2,-4) {}; 
    \node[vert, fill=ourblue] (B5) at (0,-4) {};
    \node[vert, fill=lipicsyellow] (B6) at (2,-4) {};    
    \node[vert, fill=lipicsyellow] (B7) at (-2,-6) {}; 
    \node[vert, fill=ourblue] (B8) at (0,-6) {};
    \node[vert, fill=ourblue] (B9) at (2,-6) {};    
    \node[vert, fill=ourblue] (B10) at (-2,-8) {}; 
    \node[vert, fill=lipicsyellow] (B11) at (0,-8) {};
    \node[vert, fill=lipicsyellow] (B12) at (2,-8) {};
    
    \node[vert, fill=lipicsyellow] (C) at (-3,-7) {};
    
    \path[line width=2pt]  (A1) edge (A3);
    \path[line width=2pt]  (A4) edge (A3);
    \path[line width=2pt]  (A4) edge (A2);
    \draw (A1) -- (U1);
    \draw (A1) -- (U2);
    \draw (A1) -- (U3);
    \draw (A1) -- (U4);
    \draw (A1) -- (U5);
    \draw (A2) -- (U1);
    \draw (A2) -- (U2);
    \draw (A2) -- (U3);
    \draw (A2) -- (U4);
    \draw (A2) -- (U5);
    
    \path  (B7) edge (B4);
    \path  (B8) edge (B5);
    \path  (B9) edge (B6);
    
    \path[line width=2pt]  (B7) edge (B10);
    \path[line width=2pt]  (B8) edge (B11);
    \path[line width=2pt]  (B9) edge (B12);  
     
    \path[line width=2pt]  (B7) edge (B5);
    \path[line width=2pt]  (B8) edge (B6);
    \path[line width=2pt]  (B9) edge (B4); 
     
     \path[line width=2pt, bend left=60,draw=white]  (C) edge (A1);  
    
    \path (B1) edge (B2);
    \path (B2) edge (B3);
    \path[bend left]  (B3) edge (B1);
    
    \path (B10) edge (B11);
    \path (B11) edge (B12);
    \path[bend left]  (B12) edge (B10);

    \end{tikzpicture}}
    \caption{Type 3 snapshot.}\label{fig:reductionvc:sub2}
    \end{subfigure}
\begin{subfigure}{0.24\textwidth}\centering
\scalebox{.8}{
    \begin{tikzpicture}[line width=1pt, scale=.45]
    \clip (-5,-9) rectangle (4, 5); 
    \node[vert, fill=ourblue] (A1) at (-1,2) {}; 
    \node[vert, fill=lipicsyellow] (A2) at (1,2) {};
    \node[vert, fill=lipicsyellow] (A3) at (-1,4) {};
    \node[vert, fill=ourblue] (A4) at (1,4) {};

    \node[vert2,fill=ourblue] (U1) at (-3,0) {};
    \node[vert2,fill=lipicsyellow] (U2) at (-2,0) {};
    \node (U22) at (-1,0) {$\ldots$};
    \node[vert2,fill=lipicsyellow] (U3) at (0,0) {};
    \node[vert2,fill=lipicsyellow] (U4) at (1,0) {};
    \node (U42) at (2,0) {$\ldots$};
    \node[vert2,fill=ourblue] (U5) at (3,0) {};
    
    \node[vert] (B1) at (-2,-2) {}; 
    \node[vert] (B2) at (0,-2) {};
    \node[vert] (B3) at (2,-2) {};    
    \node[vert, fill=lipicsyellow] (B4) at (-2,-4) {}; 
    \node[vert, fill=lipicsyellow] (B5) at (0,-4) {};
    \node[vert, fill=lipicsyellow] (B6) at (2,-4) {};    
    \node[vert, fill=ourblue] (B7) at (-2,-6) {}; 
    \node[vert, fill=ourblue] (B8) at (0,-6) {};
    \node[vert, fill=ourblue] (B9) at (2,-6) {};    
    \node[vert] (B10) at (-2,-8) {}; 
    \node[vert] (B11) at (0,-8) {};
    \node[vert] (B12) at (2,-8) {};
    
    \node[vert, fill=lipicsyellow] (C) at (-3,-7) {};
    
    \path[line width=2pt]  (A1) edge (A2);
    \draw (A1) -- (U1);
    \draw (A1) -- (U2);
    \draw (A1) -- (U3);
    \draw (A1) -- (U4);
    \draw (A1) -- (U5);
    \draw (A2) -- (U1);
    \draw (A2) -- (U2);
    \draw (A2) -- (U3);
    \draw (A2) -- (U4);
    \draw (A2) -- (U5);
    
    \path  (B7) edge (B4);
    \path  (B8) edge (B5);
    \path  (B9) edge (B6);
     
     \path[line width=2pt, bend left=60,draw=white]  (C) edge (A1);  
    
    \path (B1) edge (B2);
    \path (B2) edge (B3);
    \path[bend left]  (B3) edge (B1);
    
    \path (B10) edge (B11);
    \path (B11) edge (B12);
    \path[bend left]  (B12) edge (B10);

    \end{tikzpicture}}
    \caption{Type 4 snapshot.}\label{fig:reductionvc:sub3}
    \end{subfigure}
\end{center}
\caption{Illustration of the reduction from \ONEINTHREESAT\ to \SWTempColoring\ of the proof of \cref{thm:underlyingVChard}. The vertex numbering in the description of the construction corresponds to a row-wise numbering from top-left to bottom-right. The first two rows correspond to vertices $u_1$ to $u_4$. The third row corresponds to vertices $v_1$ to $v_n$. The remaining rows correspond to vertices $w_1$ to $w_{13}$. Thin edges appear in all snapshots. Thick edges never appear consecutively and hence need to be colored properly. Red edges correspond clauses. The colors of the vertices correspond to the \propDeltaTempColoring\ constructed in the proof of \cref{thm:underlyingVChard}.}
\label{fig:reductionunderlyingvc}
\end{figure}

\emph{Construction:} In the construction, we classify the snapshots of the
constructed temporal graph by the remainders of their time slots when divided by
four. This gives us \emph{type~1}, \emph{type~2}, \emph{type~3}, and \emph{type~4} snapshots,
where type~4 snapshots are the ones with a time slot that is divisible by four
and the other type numbers correspond to the remainders of the time slots.
We start by adding four vertices $u_1$, $u_2$, $u_3$, and $u_4$ to $G$. In
snapshots of type~1 or type~3 we activate edges~$\{u_1,u_2\}$, $\{u_1,u_3\}$,
and $\{u_2,u_4\}$. In snapshots of type~2 or type~4 we activate
edge~$\{u_3,u_4\}$. For each variable $x_i$ we add a vertex~$v_i$. 
We connect each of $u_3$ and $u_4$ to all $v_i$ in all snapshots.
Next, we add 13 further vertices $w_1, w_2, \ldots, w_{13}$ to $G$.
In all snapshots we pairwise connect $w_1, w_2,$ and $w_3$; pairwise connect
$w_{11}, w_{12},$ and $w_{13}$; and activate edges $\{w_4,w_7\}$, $\{w_5,w_8\}$, and $\{w_6,w_9\}$.
In
snapshots of type~2 we activate edges $\{w_1,w_4\}$, $\{w_2,w_5\}$,
$\{w_3,w_6\}$, $\{u_3, w_{10}\}$, $\{w_7, w_{10}\}$, $\{w_8, w_{10}\}$, and
$\{w_9, w_{10}\}$ (see \cref{fig:reductionvc:sub1}). In snapshots of type~3 we
activate edges $\{w_4, w_9\}$, $\{w_5, w_7\}$, $\{w_6, w_8\}$, $\{w_7,
w_{11}\}$, $\{w_8, w_{12}\}$, and $\{w_9, w_{13}\}$ (see
\cref{fig:reductionvc:sub2}).
Lastly, let $x_{i_1}$, $x_{i_2}$, and $x_{i_3}$ be the three variables
contained in clause $c_j$. Then we activate edges $\{v_{i_1},w_1\}$,
$\{v_{i_2},w_2\}$, and $\{v_{i_3},w_3\}$ in snapshot $4j -2$ (see red
edges in \cref{fig:reductionvc:sub1}). Note that snapshot $4j -2$ has
type~2.

It is easy to check that this can be done in polynomial time. Note that vertices $u_1, u_2, u_3, u_4, w_1, w_2, \ldots, w_{13}$ form a vertex cover in $G$.
We are ready now to prove that the \ONEINTHREESAT\ instance $I$ is a
yes-instance if and only if $(G,\lambda)$ admits a proper 2-SW 2-temporal coloring.

\emph{($\Rightarrow$):} Assume we are given a yes-instance of \ONEINTHREESAT\
with a satisfying assignment. We show that the constructed instance of
\SWTempColoring\ is also a yes-instance by presenting a
\propDeltaTempColoring\ with two colors. Let
blue and yellow be the two colors we use. We always color $u_1$ and $u_4$
yellow and $u_2$ and $u_3$ blue.
If variable $x_i$ is set to true in the satisfying assignment, we color $v_i$
yellow in snapshots of type~1 and~3 and blue in snapshots of type~2 and~4. If
variable~$x_i$ is set to false in the satisfying assignment, we color~$v_i$
yellow in snapshots of type~2 and~4 and blue in snapshots of type~1 and~3.
Vertex~$w_{10}$ is always colored yellow.

Let clause $c_j$ be satisfied by its $s$-th variable (note that $1\le s\le 3$). We describe how to color vertices $w_1, \ldots, w_9$ in snapshot $4\cdot j -2$. Vertex~$w_s$ is colored yellow. Vertices in $\{w_1,w_2,w_3\}\setminus\{w_s\}$ are colored blue. Vertex~$w_{s+3}$ is colored blue. Vertices in $\{w_4,w_5,w_6\}\setminus\{w_{s+3}\}$ are colored yellow. Vertices $w_7$, $w_8$, and $w_9$ are colored blue. We further describe how to color vertices $w_4, \ldots, w_{13}$ in snapshot~$4\cdot j -1$. Note that $w_{10}$ is already colored yellow. We color vertex $w_{((s+1)\bmod 3)+4}$ blue and vertices in $\{w_4,w_5,w_6\}\setminus\{w_{((s+1)\bmod 3)+4}\}$ yellow. We color vertex $w_{(s\bmod 3)+7}$ yellow and vertices in $\{w_7,w_8,w_9\}\setminus\{w_{(s\bmod 3)+7}\}$ blue. We color vertex $w_{(s\bmod 3)+11}$ blue and vertices in $\{w_{11},w_{12},w_{13}\}\setminus\{w_{(s\bmod 3)+11}\}$ yellow. 

In all snapshots of type~1 and~4 we color vertices $w_4$, $w_5$, and $w_6$ yellow and vertices $w_7$, $w_8$, and $w_9$ blue. The coloring scheme so far is depicted in \cref{fig:reductionunderlyingvc}. Note that the colors of some vertices in some snapshots are not specified yet. These are the white vertices in \cref{fig:reductionunderlyingvc}. All these vertices belong to triangles, hence we can color them in a way that each triangle has one monochromatic edge. We choose as the edge that should remain monochromatic an edge that is properly colored in both adjacent snapshots. Such an edge always exists since all these triangles are also triangles in the adjacent snapshots and a triangle is never colored completely monochromatic.

\emph{($\Leftarrow$):} Assume the constructed instance of \SWTempColoring\ is a
yes-instance and we have a \propDeltaTempColoring\ with two colors. We show
that the given instance of \ONEINTHREESAT\ is also a yes-instance by
constructing a satisfying assignment. We claim that the following yields a satisfying assignment. For every variable~$x_i$, if edge~$\{u_3, v_i\}$ is colored properly in the first snapshot, we set~$x_i$ to true, otherwise we set~$x_i$ to false.

First we argue, that if an edge~$\{u_3, v_i\}$ is colored properly in the first
snapshot for some~$1\le i\le n$ then it is also colored properly in every odd
snapshot (that is, every snapshot of type~1 and~3). Furthermore, the edge is
colored monochromatically in every even snapshot (that is, every snapshot of
type~2 and~4). Analogously, if an edge~$\{u_3, v_i\}$ is colored
monochromatically in the first snapshot for some $1\le i\le n$  then it is also
colored monochromatically in every odd snapshot and colored properly in every
even snapshot. This follows from an easily verifiable fact that in every \propDeltaTempColoring\
vertex~$u_3$ is colored different from vertex~$u_4$ in every snapshot. It follows that if an edge~$\{u_3, v_i\}$ is colored
monochromatically in a snapshot~$t$ for some~$1\le i\le n$ and~$1\le t\le T-1$
then~$\{u_3, v_i\}$ needs to be colored properly in snapshot $t+1$ meaning
that~$\{u_4, v_i\}$ is colored monochromatically in snapshot $t+1$.
Symmetrically, if an edge~$\{u_4, v_i\}$ is colored monochromatically in a
snapshot~$t$ for some~$1\le i\le n$ and~$1\le t\le T-1$ then~$\{u_4, v_i\}$
needs to be colored properly in snapshot $t+1$ meaning that~$\{u_3, v_i\}$ is
colored monochromatically in snapshot $t+1$.

Now we are ready to argue that each clause of the \ONEINTHREESAT\ instance is
satisfied. To see this we first take a look at snapshots of type~3. Note that
the triangle consisting of vertices $w_{11}$, $w_{12}$, and $w_{13}$ has
exactly one monochromatic edge. It cannot have three since not all three edges
can be colored properly in the adjacent snapshots. This means that exactly two
out of the three vertices $w_{11}$, $w_{12}$, and $w_{13}$ have the same color. It
follows that exactly two out of the three vertices $w_7$, $w_8$, and $w_9$ have
the same color, since the edges $\{w_7, w_{11}\}$, $\{w_8, w_{12}\}$, and
$\{w_9, w_{13}\}$ need to be colored properly. It is easy to check that this
implies that exactly two out of the three edges $\{w_4,w_7\}$, $\{w_5,w_8\}$,
and $\{w_6,w_9\}$ are colored monochromatically, since edges $\{w_4, w_9\}$,
$\{w_5, w_7\}$, and $\{w_6, w_8\}$ need to be colored properly.

Now we take a look at snapshots of type~2. From the last paragraph follows that
at most one out of the three edges $\{w_4,w_7\}$, $\{w_5,w_8\}$, and
$\{w_6,w_9\}$ is colored monochromatically. Since edges $\{u_3, w_{10}\}$,
$\{w_7, w_{10}\}$, $\{w_8, w_{10}\}$, and $\{w_9, w_{10}\}$ need to be colored
properly we have that vertices $w_7$, $w_8$, and $w_9$ have the same color as
vertex $u_3$. It follows that at most one of the vertices $w_4$, $w_5$, and
$w_6$ is colored in the same color as $u_3$. Since vertices $w_1$, $w_2$, and
$w_3$ form a triangle and edges~$\{w_1,w_4\}$, $\{w_2,w_5\}$, and $\{w_3,w_6\}$
need to be colored properly it follows that exactly one out of the three
vertices $w_1$, $w_2$, and $w_3$ is colored differently from $u_3$. Recall that
$w_1$, $w_2$, and $w_3$ are connected to vertices $v_{i_1}$, $v_{i_2}$, and
$v_{i_3}$ corresponding to the three variable $x_{i_1}$, $x_{i_2}$, and
$x_{i_3}$ that are contained in the clause
that corresponds to the snapshot. The connecting edges need to be colored
properly. It follows that exactly one of the edges $\{u_3,v_{i_1}\}$,
$\{u_3,v_{i_2}\}$, and $\{u_3,v_{i_3}\}$ is colored monochromatically and the
other two are colored properly. It follows that in the first snapshot, exactly
one of the edges $\{u_3,v_{i_1}\}$, $\{u_3,v_{i_2}\}$, and $\{u_3,v_{i_3}\}$ is
colored properly and the other two are colored monochromatically. This means we
set exactly one of the three variables to true and the clause is satisfied.
\end{proof}

Next, we consider a canonical optimization version of \SWTempColoring, which we call \MinSWTempColoring, where the goal is to minimize the number of colors~$k$.
Using \Cref{thm:fpt-size-n}, we provide an FPT-approximation algorithm with an additive error of one where the parameter is the vertex cover number of the underlying graph. Considering that we cannot hope for an exact FPT algorithm for parameter ``vertex cover number of the underlying graph'' unless P $=$ NP (cf.~\cref{thm:underlyingVChard}), this is the best we can get from a classification standpoint.

\begin{theorem}\label{thm:approx}
\MinSWTempColoring\ admits a linear time FPT-approximation algorithm with an additive error of one when parameterized by the vertex cover number of the underlying graph.
\end{theorem}
\begin{proof}
First, we compute a minimum vertex cover of the underlying graph. Note that this can be done in linear FPT time~\cite{CyganFKLMPPS15}. 
We use the algorithm of \Cref{thm:fpt-size-n} to compute a minimum \propDeltaTempColoring\ for the temporal graph induced by the vertex cover vertices\footnote{Note that the algorithm presented in \Cref{thm:fpt-size-n} solves the decision version of \SWTempColoring\ while we want to solve the minimization problem here. This can be done by trying out all values for the number $k$ of colors between one and the size of the vertex cover of the underlying graph.}. Note that this computation takes linear FPT time and the number of colors used is clearly a lower bound for the minimum number of colors necessary to properly color the whole temporal graph. We color the remaining vertices with a fresh color. This clearly gives a \propDeltaTempColoring\ for the whole temporal graph that uses at most one extra color compared to the optimum.
\end{proof}

\section{Conclusion}

In this paper we introduced and rigorously studied a natural temporal extension of the classical graph coloring problem, 
called \textsc{Sliding Window Temporal Coloring}. 
We showed that \textsc{Sliding Window Temporal Coloring}  is NP-hard even under severe restrictions. On the positive side, we provided a linear time FPT-algorithm for parameter number $n$ of vertices and a linear time FPT-approximation algorithm for parameter vertex cover number of the underlying graph with an additive error of one. We leave as an open question whether for the latter, we can replace vertex cover number by a structurally smaller parameter. A natural extension of our problem would be to impose a restriction on the number of color reassignments of vertices.

\bibliography{temp-coloring-bibliography}
\bibliographystyle{abbrvnat}

\end{document}